\documentclass[11pt, journal, draftcls, onecolumn]{IEEEtran}

\ifCLASSINFOpdf
\else
\fi
\usepackage{amsmath}
\usepackage{amssymb}
\usepackage{mathtools}
\usepackage{algorithmic}
\usepackage[mathscr]{euscript}
\usepackage{graphicx,subfigure}
\usepackage{amsthm}
\usepackage[scaled]{helvet} 
\usepackage{cases}
\usepackage{cite}
\usepackage{algorithmic}
\usepackage{epstopdf}
\usepackage{color}
\usepackage[lined,boxed,commentsnumbered,ruled,vlined]{algorithm2e}
\usepackage{makeidx}
\usepackage[]{algorithm2e}

\newcommand{\nonr}{\nonumber}
\newcommand{\revcol}{\color{black}}  
\newcommand{\revcoltwo}{\color{black}}  
  
\newcommand{\beq}{\begin{equation}} 
\newcommand{\eeq}{\end{equation}} 
\newcommand{\beqa}{\begin{eqnarray}}
\newcommand{\eeqa}{\end{eqnarray}}

\newcommand{\vsig}{\mathbf{x}}  

\newcommand{\smpMat}{\mathbf{A}}

\newcommand{\mT}{\mathbf{T}} 
\newcommand{\sS}{\mathscr{S}}

\newcommand{\mG}{\mathbf{G}}

\newcommand{\Exp}{\mathbb{E}}

\newcommand{\bR}{\mathbb{R}}

\newcommand{\mP}{\mathbf{P}}
\newcommand{\cN}{\mathcal{N}}
\newcommand{\cT}{\mathcal{T}}
\newcommand{\cC}{\mathcal{C}}

\newcommand{\vy}{\mathbf{y}}

\newcommand{\vu}{\mathbf{u}}
\newcommand{\vs}{\mathbf{s}}

\newcommand{\vt}{\mathbf{t}}
\newcommand{\vz}{\mathbf{z}}
\newcommand{\vw}{\mathbf{w}}

\newcommand{\mD}{\mathbf{D}}

\newcommand{\mX}{\mathbf{X}}
\newcommand{\mI}{\mathbf{I}}

\newcommand{\mZero}{\mathbf{0}}

\newcommand{\Phib}{\mbox{\boldmath$\Phi$\unboldmath}}

\newcommand{\Omegab}{\mbox{\boldmath$\Omega$\unboldmath}}

\newcommand{\varepsb}{\mbox{\boldmath$\varepsilon$\unboldmath}}

\newcommand{\trace}{\mathrm{trace}}

\newcommand{\Nb}{N_b}

\newcommand{\diag}{\mathrm{diag}}

\newtheorem{thm}{Theorem}
\newtheorem{prop}{Proposition}
\newtheorem{cor}{Corollary}

\newcommand{\mH}{\mathbf{H}}

\newcommand{\mU}{\mathbf{U}}

\newcommand{\mA}{\mathbf{A}}

\newcommand{\mV}{\mathbf{V}}

\newcommand{\vx}{\mathbf{x}}
\newcommand{\tr}{T}

\newcommand{\nk}{L}

\newcommand{\pfa}{P_{\mbox{\scriptsize FA}}}
\newcommand{\pdet}{P_{\mbox{\scriptsize D}}}

\newcommand{\Qfun}{\mathbb{Q}}

\newcommand{\pr}{\pi}

\newcommand{\vone}{\mathbf{1}}
\newcommand{\vzero}{\mathbf{0}}

\newcommand{\cH}{\mathcal{H}}
\newcommand{\cR}{\mathcal{R}}
\newcommand{\cChi}{\mathcal{X}}
\newcommand{\cM}{\mathcal{M}}

\newcommand{\mLambda}{\mbox{\boldmath$\Lambda$\unboldmath}}

\newcommand{\mPsi}{\mbox{\boldmath$\Psi$\unboldmath}}

\newcommand{\veps}{\mbox{\boldmath$\epsilon$\unboldmath}}

\newcommand{\bP}{\mathbb{P}}

\newcommand{\mPhi}{\mbox{\boldmath$\Phi$\unboldmath}}
\newcommand{\vphi}{\mbox{\boldmath$\phi$\unboldmath}}
\newcommand{\vfi}{\mbox{\boldmath$\psi$\unboldmath}}
\newcommand{\vwp}{\vw^{\phi}}
\newcommand{\sI}{\mathcal{I}}
\newcommand{\vxx}{\underline{\vx}}

\begin{document}
\title{Compressive Detection of Random Subspace Signals}
\author{\IEEEauthorblockN{Alireza Razavi, Mikko Valkama, and Danijela Cabric}
\thanks{A. Razavi and M. Valkama are with the Department of Electronics and Communications Engineering, Tampere University of Technology (TUT), Tampere, Finland. emails: alireza.razavi@tut.fi, mikko.e.valkama@tut.fi.}
\thanks{D. Cabric is with Cognitive Reconfigurable Embedded Systems Lab (CORES), University of California Los Angeles (UCLA), CA. email:danijela@ee.ucla.edu.}
}

\maketitle

\begin{abstract}
The problem of compressive detection of random subspace signals is studied. We consider signals modeled as $\vs = \mH \vx$ where $\mH$ is an $N \times K$ matrix with $K \le N$ and $\vx \sim \cN(\vzero_{K,1},\sigma_x^2 \mI_K)$. We say that signal $\vs$ lies in or leans toward a subspace if the largest eigenvalue of $\mH \mH^\tr$ is strictly greater than its smallest eigenvalue. We first design a measurement matrix $\mPhi=[\mPhi_s^\tr,\mPhi_o^\tr]^\tr$ comprising of two sub-matrices $\mPhi_s$ and $\mPhi_o$ where $\mPhi_s$ projects the signal to the strongest left-singular vectors, i.e., the left-singular vectors corresponding to the largest singular values, of subspace matrix $\mH$ and $\mPhi_o$ projects it to the weakest left-singular vectors. We then propose two detectors which work based on the difference in energies of the samples measured by the two sub-matrices $\mPhi_s$ and $\mPhi_o$ {\revcoltwo and provide theoretical proofs for their optimality}. Simplified versions of the proposed detectors for the case when the variance of noise is known are also provided. Furthermore, we study the performance of the detector when measurements are imprecise and show how imprecision can be compensated by employing more measurement devices. The problem is then re-formulated for the generalized case when the signal lies in the union of a finite number of linear subspaces instead of a single linear subspace. Finally, we study the performance of the proposed methods by simulation examples.
\end{abstract}


\section{Introduction}
\label{sec:intro}
\IEEEPARstart{T}{he} topic of Compressive Sensing (CS), where the effort is to draw inferences about the sparse signals based on compressive samples, has been one of the most attractive topics in the area of signal processing for more than a decade. 
Compressive Sensing states that if a signal is sparse in some known bases, then it is possible to reconstruct it from a relatively few linear projections \cite{donoho06, candes06, candes08}. While most of the efforts in the area of CS has been devoted to the problem of estimation of a sparse signal from possibly noisy compressive measurements, there have been also some efforts in accomplishing other signal processing tasks such as Compressive Detection \cite{davenport10, davenport06, haupt07detection, paredes09, wang08, zahedi12CD, azizyan12, rao12detection, arias12}, Compressive Classification \cite{haupt06, duarte07, reboredo13,davenport10, davenport07}, etc. The goal of this paper is to study the problem of detecting a random subspace signal from its compressive measurements.

Signal detection \cite{kay_volII} is an important task in Statistical Signal Processing where the goal is to decide on the presence/absence of a signal rather than estimating it. Signal Detection has applications in many areas of engineering and science, including wireless communications, radar and sonar, bioinformatics, etc. The literature on detection from classical uncompressed samples is very vast and well-studied; see, e.g., \cite{kay_volII, poor94detection, scharf-book} and references therein. With the advent of compressive sampling methods, like many other fields, there have been additional efforts to tailor the existing signal detection techniques to the case where the samples are taken at compressive rates. The first criterion that a signal has to meet to qualify for being detectable based on compressive samples is to be {\it sparse} in some known basis, or in a broader sense to lie in a {\it low-dimensional subspace} of the higher-dimensional ambient space. {\revcol The applications of such compressive detection techniques are in the scenarios where uncompressed sampling is difficult, expensive, or even harmful, but because of the existence of a structure in the signal we are still able to sample the signal in compressive rate and then detect it based on the compressive samples. Examples of such scenarios are compressive spectrum sensing of OFDM signals \cite{razavi13,razaviSP14}, radar sensing \cite{baransky14,yu11,ahmed10,bar14}, hyperspectral imaging \cite{paris13}, ultrasound imaging \cite{wagner12}, etc. We remark here that the hardware implementation of specific compressive samplers is beyond the scope of this paper and we only focus on finding the optimal mathematical model or framework of such samplers, however examples of hardware implementations in various applications can be found in \cite{duarte08single, dixon12, itzhak13, wagner12, baransky14, mishali10, xampling, tropp10}.
}

In this paper we address the problem of compressive detection of signals that lie on or close to a low dimensional linear subspace. We start with the simplest case when the signal is drawn from a known low-dimensional subspace and propose two tests for detection of signal from compressive measurements. The main advantage of the proposed algorithms over the existing methods is that they do not need the knowledge of the noise variance to work. However, we also provide the simplified versions of the proposed methods for the case when the noise variance is known. We then study the effect of imprecise measurements and show that the imprecision can be compensated by employing more measurement devices. Then the proposed method is re-formulated for the case when the signal is lying on the union of a finite number of known linear subspaces with each of subspaces having a certain probability of being the subspace on which the signal truly lies. In other words, in this case, instead of having the exact knowledge of the subspace, we know the set of all subspaces from which the signal may emerge together with the corresponding probabilities that the signal may lie on each of them. We will introduce the optimum sampling strategy in this case and show that although the performance falls by increasing the number of possible subspaces, it remains reasonable as long as the cardinality of the set is much smaller than the ambient dimension.


{\bf Related work:} In \cite{davenport10, davenport06, haupt07detection, arias12} authors studied the problem of compressive detection, but they formulated the problem for deterministic agnostic signals rather than random subspace signals. Besides, their methods did not take into consideration the unknown noise variance scenario. The method in \cite{paredes09} designs a matched subspace detector for subspace signals.  The method assumes that the exact knowledge of signal is available for detector design. They extend their work to unknown signals in \cite{wang08} but like the previous ones they need the variance of noise to design the compressive detector. The work of \cite{azizyan12} provided upper bounds for probabilities of false alarm and mis-detection for deterministic signals while the variance of noise is again assumed known. Finally, in \cite{rao12detection} authors studied the problem of compressive detection of random subspace signals but similar to the previously mentioned works the noise variance was assumed known and the signal was assumed agnostic. 

{\bf Notations and Mathematical Preliminaries:} Throughout this paper, all quantities are assumed real-valued while matrices and vectors are denoted by capital and small boldface letters, respectively. $=$ denotes the equality and $\triangleq$ denotes the definition.
$\Exp_X$ is reserved for statistical expectation operator with respect to the random variable $X$, and $\otimes$ denotes the Kronecker product.
$\mI_P$,  $\vone_{P,Q}$ and $\vzero_{P,Q}$ represent
$P \times P$ identity matrix, and $P \times Q$ all-one and all-zero matrices, respectively. The set of real numbers is denoted by $\bR$, the set of real-valued $M \times 1$ vectors is denoted by $\bR^M$, the set of real-valued $M \times N$ matrices is denoted by $\bR^{M \times N}$, and for an arbitrary matrix $\smpMat \in \bR^{M \times N}$, $[\smpMat]_{i,j}$ denotes the $(i,j)$-th entry of the matrix. For a set $\cM$, cardinality of the set is shown as $|\cM|$.
For a matrix $\mA$ we show its $i$-th eigenvalue by $\lambda_i(\mA)$. We also use the notations $\lambda_{\mathrm{max}}(\mA)$ and $\lambda_{\mathrm{min}}(\mA)$ for denoting its maximum and minimum eigenvalues, respectively. For a vector $\vx$, its sub-vector containing entries from $i$ to $j>i$ is denoted by $\vx\Big|_{i:j}$.

{\bf Paper Organization:} The rest of this paper is organized as follows: first, in Section \ref{sec:sysmodel} we explain the basic model of a linear subspace signal and we introduce the sampling strategy that we choose for our compressive detector. Then based on the introduced sampling strategy, in Section \ref{sec:2} we introduce two compressive detectors for detection of the signal based on compressive samples and provide theoretical results about their optimality. We also briefly study the simplification when the noise variance is known as well as the case of subspace interference. Next, in Section \ref{sec:nonideal}, we study the effect of imprecise measurements where, e.g. because of hardware design limitations, the measurement matrix may deviate from the ideal one introduced in Section \ref{sec:2}.  In Section \ref{sec:union}, we extend our design to the case where instead of knowing the exact linear subspace on which signal lies, we have only more coarse knowledge about the signal location in the ambient space in the form of knowing all the possibilities for the true subspace together with their corresponding probabilities.  The performances of the proposed methods are studied through computer simulations  in Section \ref{sec:simulation}. Finally we conclude the paper in Section \ref{sec:conclusion}.


\section{System Model}
\label{sec:sysmodel}
{\revcol
Consider first a noiseless subspace signal of the form
\beq
\vs[n] \triangleq \mH \vx[n],~~n=1,\ldots,\Nb,
\label{eq:sysModel}
\eeq
where $\vx[n] \sim \cN(\mZero_{K,1},\sigma_x^2 \mI_K)$ and $\mH$  is an $N \times K$ full-rank deterministic matrix with $K \le N$.
We denote the $N$ nonnegative eigenvalues of $\mH \mH^\tr$ by $\rho_1^2 \ge \rho_{2}^2 \ge \ldots \ge \rho_N^2 \ge 0$ and we further assume that
\beq
\rho_1^2>\rho_N^2.
\label{ineq:condnum}
\eeq
Inequality (\ref{ineq:condnum}) implies that the signal energy is not distributed uniformly over all dimensions, or in other words, the signal lies in (or at least leans toward) some subspace of the ambient space.

In a classical uncompressed scenario, by processing a sequence of noisy observations $\vy[n], n=1,\ldots,\Nb$, the signal detection problem refers to the following hypothesis testing
\beqa
\left\{ \begin{array}{lll} \cH_0: \vy[n] = \vw[n],&~&n=1,2,\ldots,\Nb, \\ 
\cH_1: \vy[n] =  \vs[n] + \vw[n],&~&n=1,2,\ldots,\Nb, \end{array} \right.
\label{eq:test0}
\eeqa
in which $\vw[n]\sim \cN(\mZero_{N,1},\mI_N)$ refers to observation noise. However, as mentioned in the previous section, obtaining uncompressed observations is sometimes difficult or expensive, and therefore we have to perform the signal detection task through some noisy compressive observations $\vz[n] \triangleq [z_1[n],z_2[n],\ldots,z_M[n]]^T, n=1,2,\ldots,\Nb$, rather than unavailable uncompressed samples $\vy[n], n=1,2,\ldots,\Nb$. Here $z_m[n],~m=1,2,\ldots,M$ denotes the output of $m$-th compressive measurement unit. Thus, we redefine the hypothesis testing problem as
\beqa
\left\{ \begin{array}{lll} \cH_0: \vz[n] = \vwp[n],&~&n=1,2,\ldots,\Nb, \\ 
\cH_1: \vz[n] = \mPhi \vs[n]+ \vwp[n],&~&n=1,2,\ldots,\Nb, \end{array} \right.
\label{eq:test1}
\eeqa 
where $\mPhi \triangleq [\vphi_1,\vphi_2,\ldots,\vphi_M]^T \in \bR^{M \times N}$ is compressive measurement matrix with $M \ll N$ and $\vwp[n]\triangleq[\vphi_1^T \vw_1[n],\vphi_2^T \vw_2[n],\ldots,\vphi_M^T \vw_M[n]]^T$ where $\vw_m[n] \sim \cN(\mZero_{N,1},\mI_N)$ represents the uncompressed receiver noise at the input of the $m$-th compressive measurement unit. We have used different indices for noise vectors as we assume that these measurement units work independently and posses independent noise at their inputs, which means that $\vw_{m_1}[n]$ and $\vw_{m_2}[n]$ are statistically independent for $m_1 \neq m_2$.

This problem was addressed in \cite{paredes09,wang08} in presence of noise with known variance and for deterministic signals. In this paper, we design compressive detectors without assuming prior noise knowledge. To this end, we partition the $M$ measurement devices to two sets or, in other words, we consider
 the measurement matrix $\mPhi$ to consist of two sub-matrices $\mPhi_s \in \bR^{M_1 \times N}$ and $\mPhi_o \in \bR^{M_2 \times N}$ as
 \beq
\mPhi=   \left[ \begin{array}{c}   \mPhi_s \\  \mPhi_o  \end{array} \right],
\label{measMat}
\eeq}
{\revcol 
where $M_1+M_2=M$. The idea behind having two sets of samplers, $\mPhi_s$ and $\mPhi_o$, is to get two sets of measurements 
\beq
\vz_s[n] \triangleq \vz[n]\Big|_{1:M_1},
\label{eq:vzs}
\eeq 
and 
\beq
\vz_o[n]\triangleq \vz[n]\Big|_{M_1+1:M},
\label{eq:vzo}
\eeq
whose selected statistics are identical under null hypothesis but different under alternative hypothesis. The statistics of the two sets of measurements can then be compared to decide whether we should accept the null hypothesis $\cH_0$ or reject it in favor of alternative hypothesis $\cH_1$. In the next section, we provide two designs for $\mPhi_s$ and $\mPhi_o$ that fulfill this goal and then propose detectors for each design to perform the hypothesis testing problem of (\ref{eq:test1}) based on two sets of compressive measurements $\vz_s[n]$ and $\vz_o[n], n=1,2,\ldots,\Nb$.
}



\section{Compressive Subspace Detection}
\label{sec:2}

Based on the system model of the previous section, in this section we propose two detectors for performing the hypothesis testing problem of (\ref{eq:test1}). Both detectors are composed of two sets of measurement devices, represented by  $\mPhi_s$ and $\mPhi_o$. {\revcoltwo In the sequel, we introduce these two detectors, derive their performance measures, and prove their optimality. The numerical study of the performance of the two proposed detectors will be provided in Section \ref{sec:simulation} through simulation experiments. }

\subsection{Maximally-Uncorrelated Compressive Detector}
\label{sec:maxrank}

 Assume that the SVD of matrix $\mH \in \bR^{N \times K}$ is given by
\beq
\mH = \mU \mD \mV^\tr,
\label{subspaceMat}
\eeq
where $\mathbf{U} \in \mathbb{R}^{N \times N}$ and $\mathbf{V} \in \mathbb{R}^{K \times K}$ are orthogonal matrices, and $\mathbf{D} \in \mathbb{R}^{N \times K}$ is a diagonal matrix.
Consider positive integers $M_1$ and $M_2$ such that $M=M_1+M_2$.  Let us denote the first $M_1$ columns of $\mU$ by $\mU_{s}$ and its last $M_2$ by $\mU_{o}$. We propose the following design for $M \times N$ measurement matrix $\mPhi=\left[ \begin{array}{c}   \mPhi_s \\  \mPhi_o  \end{array} \right]$ with sub-matrices $\mPhi_s$ and $\mPhi_o$ defined as
\beq
\mPhi_s \triangleq \frac{1}{\sqrt{M}} \mT_s \mU_{s}^\tr,
\label{eq:measS}
\eeq
and 
\beq
\mPhi_o \triangleq \frac{1}{\sqrt{M}} \mT_o \mU_{o}^\tr,
\label{eq:measO}
\eeq
where $\mT_s$ and $\mT_o$ are $M_1 \times M_1$ and $M_2 \times M_2$ orthogonal matrices, respectively. The term $\frac{1}{\sqrt{M}}$ in (\ref{eq:measS}) and (\ref{eq:measO}) is to guarantee that the received energy at the output of the compressive sampler is independent of the number of samplers. This makes the study of the effect of the number of samplers on the performance fair as otherwise it is obvious that the bigger is the number of samplers the better is the performance. We remark again that the $M$ rows of the measurement matrix $\mPhi$ represent $M$ measurement devices equipped to collect samples. Hereafter, we use the terms {\it measurement matrix} and {\it measurement devices} interchangeably.

The reason for choosing this design for measurement matrix is that the rows of $\mPhi_s$ span the $M_1$-dimensional subspace which contains the highest amount of energy of signal $\vs=\mH \vx$, or in other words, is the $M_1$-dimensional subspace along which the variance of $\vs$ is maximized. Furthermore, the rows of $\mPhi_o$ span the $M_2$-dimensional subspace which contains the lowest amount of energy of signal $\vs$, or in other words, is the $M_2$-dimensional subspace along which the variance of $\vs$ is minimized. The difference between the energy of signal taken at the output of these two sub-matrices can then be exploited for signal detection. The following theorem, justifies the above discussion.

\begin{thm}
The measurement sub-matrix $\mPhi_s$ in (\ref{eq:measS}) is the solution to the following optimization problem
\beq
\arg \max_{\mPhi_s} \Exp_{\vx,\vw}( \| \vz_s\|_2^2),~~\mathrm{s.t.}~~M \mPhi_s \mPhi_s^\tr=\mI_{M_1},
\label{eq:th0eq1}
\eeq
In other words, the rows of $\mPhi_s$ in (\ref{eq:measS}) represent the set of $M_1$ {\it uncorrelated} measurement devices that maximize the expected value of total energy (or the total variance) in their outputs.

Correspondingly, the measurement sub-matrix $\mPhi_o$ in (\ref{eq:measO}) is the solution to the following optimization problem
\beq
\arg \min_{\mPhi_o}  \Exp_{\vx,\vw}( \| \vz_o\|_2^2),~~\mathrm{s.t.}~~M \mPhi_o \mPhi_o^\tr=\mI_{M_2},
\label{eq:th0eq2}
\eeq
In other words, the rows of $\mPhi_o$ in (\ref{eq:measO}) represent the set of $M_2$ {\it uncorrelated} measurement devices that minimize the expected value of total energy (or the total variance) in their outputs.
\label{thm0}
\end{thm}
\begin{proof}
We first notice that
\beqa
\Exp_{\vx,\vw}( \| \vz_s\|_2^2) &=&  \Exp_{\vx} (\vx^\tr \mH^\tr \mPhi_s^\tr \mPhi_s \mH \vx) + \frac{M_1 \sigma_0^2}{M} \nonr \\
&=& \sigma_x^2 \trace(\mH \mH^\tr \mPhi_s^\tr \mPhi_s) + \frac{M_1 \sigma_0^2}{M} 
\label{eq:th1proof1}
\eeqa
Noticing that the second term in the right-hand side of (\ref{eq:th1proof1}) is independent of measurement design and $\mH^\tr \mH$ and $\mPhi_s^\tr \mPhi_s$ are both symmetric matrices and that $M_1 \mPhi_s^\tr \mPhi_s$ in (\ref{eq:th0eq1}) has $M_1$  eigenvalues equal to 1 and $N-M_1$ eigenvalues equal to zero,
the proof of (\ref{eq:measS}) is easily concluded by employing von Neumann trace theorem\cite[P 11.4.5]{raorao},\cite[Theorem 6.77]{seber08}. The proof for (\ref{eq:measO}) is similar.
\end{proof}
Hereafter, we call $\mPhi_s$ and $\mPhi_o$ as {\it max-energy sampler} and {\it min-energy sampler}, respectively, and call the measurement design in (\ref{eq:measS}) and (\ref{eq:measO}) {\it Maximally-Uncorrelated Compressive Detector} as it is the optimum design when the measurements are uncorrelated.

Next, based on the above measurement design we introduce a test statistic for carrying out the hypothesis testing problem of (\ref{eq:test1}). The following theorem states this test and derives its performance measures.
\begin{thm}
\label{theorem1}
Consider the following test for performing the hypothesis testing problem given in (\ref{eq:test1}): 
\beq
\cT = \frac{\sum\limits_{n=1}^{\Nb} \vz_s[n]^\tr \vz_s[n]}{\sum\limits_{n=1}^{\Nb} \vz_o[n]^T \vz_o[n]}\overset{\cH_1}{\underset{\cH_0}{\gtrless}} \gamma,
\label{eq:testStat}
\eeq
where $\vz_s[n]$ and $\vz_o[n]$ are as in (\ref{eq:vzs}) and (\ref{eq:vzo}), respectively. The Probability of False Alarm, $\pfa$, for this test is computed as
\beq
\pfa=\Qfun_{F(M_1 \Nb,M_2 \Nb)}(\gamma),
\label{eq:pfa}
\eeq
and the Probability of Detection, $\pdet$, is bounded as
\beq
P_{D,lb} \le \pdet \le P_{D,ub}
\label{eq:pd}
\eeq
where $P_{D,lb}$ and $P_{D,ub}$ are, respectively, the lower bound and the upper bound for the probability of detection which will take the following formulas
\beq
P_{D,lb}= \Qfun_{F(M_1 \Nb,M_2 \Nb)}({\eta_{lb}}\gamma),
\label{eq:pdlb}
\eeq
and
\beq
P_{D,ub}= \Qfun_{F(M_1 \Nb,M_2 \Nb)}({\eta_{ub}}\gamma),
\label{eq:pdub}
\eeq
where $\eta_{lb}\triangleq\frac{\sigma_0^2+\sigma_x^2 \rho_{N-M_2+1}^2}{\sigma_0^2+\sigma_x^2 \rho_{M_1}^2}$, $\eta_{ub}\triangleq\frac{\sigma_0^2+\sigma_x^2 \rho_{N}^2}{\sigma_0^2+\sigma_x^2 \rho_{1}^2}$, and $\Qfun_{F(d_1 ,d_2)}(x)$ is the {\it tail probability} of the $F$-distribution with parameters (degrees of freedom) $d_1$ and $d_2$ at point $x$.
\end{thm}
\begin{proof}
To prove the theorem, we first notice that $\vz_o[n]$ has distribution
\beq
\vz_o[n] \sim \left\{ \begin{array}{ll}   \cN(\vzero,\frac{\sigma_0^2}{M}\mI_{M_2}) & \mathrm{under~} \cH_0,  \\ \\
 \cN(\vzero,\frac{1}{M}(\sigma_0^2\mI_{M_2}+\sigma_x^2 \mT_o \mLambda_o \mT_o^\tr) )& \mathrm{under~} \cH_1, \end{array}\right.
\label{eq:vzo1}
\eeq
where $\mLambda_o$ is an $M_2 \times M_2$ diagonal matrix whose $(i,i)$-th diagonal element is $\rho_{N-M_2+i}^2$.

On the other hand, the distribution of $\vz_s$ is
\beq
\vz_s[n] \sim \left\{ \begin{array}{ll}   \cN(\vzero,\frac{\sigma_0^2}{M}\mI_{M_1}) & \mathrm{under~} \cH_0,  \\ \\
 \cN(\vzero,\frac{1}{M}(\sigma_0^2\mI_{M_1}+\sigma_x^2 \mT_s \mLambda_s \mT_s^\tr) )& \mathrm{under~} \cH_1, \end{array}\right.
\label{eq:vzs1}
\eeq
where $\mLambda_s$ is an $M_1 \times M_1$ diagonal matrix whose $(i,i)$-th diagonal element is $\rho_{i}^2$.

From  (\ref{eq:vzs1}) we have
\beqa
\frac{\sum\limits_{n=1}^{\Nb} \vz_s[n]^\tr \vz_s[n]}{\sigma_0^2/M} \sim \cChi_{M_1 \Nb}^2, \mathrm{under~}  \cH_0,
\eeqa
where $\cChi_{d}^2$ denotes chi-squared distribution with $d$ degrees of freedom \cite{kay_volII}. On the other hand, under $\cH_1$ we have
\beqa
{\sum\limits_{n=1}^{\Nb} \vz_s[n]^\tr \mA_s^{-1} \vz_s[n]} \sim \cChi_{M_1 \Nb}^2, \mathrm{under~}  \cH_1.
\eeqa
where $\mA_s\triangleq\frac{1}{M}(\sigma_0^2\mI_{M_1}+\sigma_x^2 \mT_s \mLambda_s \mT_s^\tr)$.
From Rayleigh-Ritz theorem \cite[Chapter 8]{zhang2011matrix} we know that
\beq
\lambda_{\min}(\mA_s^{-1} ) \vz_s^\tr \vz_s \le \vz_s^\tr \mA_s^{-1} \vz_s \le \lambda_{\max}(\mA_s^{-1} ) \vz_s^\tr \vz_s
\eeq
where $\lambda_{\min}(\mA_s^{-1} )=\frac{1}{\lambda_{\max}(\mA_s)}=\frac{M}{\sigma_0^2+\sigma_x^2 \rho_{1}^2}$ and $\lambda_{\max}(\mA_s^{-1} )=\frac{1}{\lambda_{\min}(\mA_s)}=\frac{M}{\sigma_0^2+\sigma_x^2 \rho_{M_1}^2}$.

Similarly, from (\ref{eq:vzo1}) we have
\beqa
\frac{\sum\limits_{n=1}^{\Nb} \vz_o[n]^\tr \vz_o[n]}{\sigma_0^2/M} \sim \cChi_{M_1 \Nb}^2, \mathrm{under~}  \cH_0, 
\eeqa
and
\beqa
{\sum\limits_{n=1}^{\Nb} \vz_o[n]^\tr \mA_o^{-1} \vz_o[n]} \sim \cChi_{M_2 \Nb}^2, \mathrm{under~}  \cH_1.
\eeqa
where $\mA_o\triangleq\frac{1}{M}(\sigma_0^2\mI_{M_2}+\sigma_x^2 \mT_o \mLambda_o \mT_o^\tr)$. We also have
\beq
\lambda_{\min}(\mA_o^{-1} ) \vz_o^\tr \vz_o \le \vz_o^\tr \mA_o^{-1} \vz_o \le \lambda_{\max}(\mA_o^{-1} ) \vz_s^\tr \vz_s
\eeq
where $\lambda_{\min}(\mA_o^{-1} )=\frac{1}{\lambda_{\max}(\mA_o)}=\frac{M}{\sigma_0^2+\sigma_x^2 \rho_{N-M_2+1}^2}$ and $\lambda_{\max}(\mA_o^{-1} )=\frac{1}{\lambda_{\min}(\mA_o)}=\frac{M}{\sigma_0^2+\sigma_x^2 \rho_{N}^2}$.

From above computations, under $\cH_0$ hypothesis we will have
\beq
\frac{\sum\limits_{n=1}^{\Nb} \vz_s[n]^\tr \vz_s[n]}{\sum\limits_{n=1}^{\Nb} \vz_o[n]^\tr \vz_o[n]} \sim F_{M_1 \Nb,M_2 \Nb},\mathrm{under~}  \cH_0.
\label{eq_h0}
\eeq
Under $\cH_1$ hypothesis we have
\beq
\frac{\frac{\sum\limits_{n=1}^{\Nb} \vz_s[n]^\tr \vz_s[n]}{(\sigma_0^2+\sigma_x^2 \rho_{1}^2)/M}}{\frac{\sum\limits_{n=1}^{\Nb} \vz_o[n]^\tr \vz_o[n]}{(\sigma_0^2+\sigma_x^2 \rho_{N}^2)/M}} \le 
\frac{\sum\limits_{n=1}^{\Nb} \vz_s[n]^\tr \mA_s^{-1} \vz_s[n]}{{\sum\limits_{n=1}^{\Nb} \vz_o[n]^\tr \mA_o^{-1} \vz_o[n]}} \sim F_{M_1 \Nb,M_2 \Nb}
\label{ineq1_h1}
\eeq
and similarly
\beq
\frac{\frac{\sum\limits_{n=1}^{\Nb} \vz_s[n]^\tr \vz_s[n]}{(\sigma_0^2+\sigma_x^2 \rho_{M_1}^2)/M}}{\frac{\sum\limits_{n=1}^{\Nb} \vz_o[n]^\tr \vz_o[n]}{(\sigma_0^2+\sigma_x^2 \rho_{N-M_2+1}^2)/M}} \ge 
\frac{\sum\limits_{n=1}^{\Nb} \vz_s[n]^\tr \mA_s^{-1} \vz_s[n]}{{\sum\limits_{n=1}^{\Nb} \vz_o[n]^\tr \mA_o^{-1} \vz_o[n]}} \sim F_{M_1 \Nb,M_2 \Nb}
\label{ineq2_h1}
\eeq
From (\ref{eq_h0}), (\ref{ineq1_h1}), and (\ref{ineq2_h1}) the proof is concluded.
\end{proof}
{\bf Remark 1}: The $\Qfun$-function of $F$-distribution is expressed as \cite{johnson95}
\beq
\Qfun_{F(d_1,d_2)}(x)=1-I_{\frac{d_1 x}{d_1 x + d_2}}(\frac{d_1}{2},\frac{d_2}{2}), 
\label{eq:Qf}
\eeq
where $I_{x}(a,b)$ is the {\it regularized incomplete Beta function} defined as
\beq
I_{x}(a,b) = \frac{B(x;a,b)}{B(a,b)}, 
\eeq
where $B(x;a,b) \triangleq \int_0^x t^{a-1} (1-t)^{b-1} dt$ is the {\it incomplete Beta function}  and $B(a,b) \triangleq B(1;a,b)$ is the {\it Beta function}. \qed

{\bf Remark 2}: From Theorem \ref{theorem1} it is clear that for a given $M_1$ and $M_2$ the proposed detector is {\it unbiased}\cite[Chapter 4]{scharf-book} if $\rho_{N-M_2+1}<\rho_{M_1}$. Furthermore, if $\rho_1^2>\rho_N^2$, then there is always a design as suggested in this section with certain $M_1$ and $M_2$ which delivers an unbiased detector. \qed

{\revcol
Theorem \ref{thm0} showed that the proposed design of (\ref{eq:measS}) and (\ref{eq:measO}) is the one that provides the largest difference between the expected energies of signals $\vz_s$ and $\vz_o$ under orthogonality condition. Then in Theorem \ref{theorem1} a test for detecting random subspace signals based on this design was proposed and the probabilities of false alarm and detection for the proposed detector were derived. An important measure in statistical hypothesis testing is to show that the proposed test is the most powerful test in the sense that it has the highest probability of detection, $\pdet$, for a fixed probability of false alarm, $\pfa$. In  the following a corollary is provided to show that among all orthogonal designs, the one proposed in (\ref{eq:measS}) and (\ref{eq:measO}) delivers the highest probability of detection for the detector in (\ref{eq:testStat}).

\begin{cor}
\label{corol0}

Among all orthogonal designs satisfying $\mPhi \mPhi^T=\frac{1}{M} \mI_{M}$ with fixed number of measurement devices $(M_1,M_2)$, the design of (\ref{eq:measS}) and (\ref{eq:measO})  provides the highest $P_{D,lb}$ and $P_{D,ub}$ for a given $\pfa$.
\end{cor}
\begin{proof}
Consider an arbitrary orthogonal measurement design as $\mPhi=[\mPhi_s^T,\mPhi_o^T]^T$, where $\mPhi_s \in \bR^{M_1 \times N}$ and $\mPhi_o \in \bR^{M_2 \times N}$. The goal is to show that if we choose them as in (\ref{eq:measS}) and (\ref{eq:measO}) then $P_{D,lb}$ and $P_{D,ub}$ are maximized for a given $\pfa$. To this end, first by taking the same steps as in the proof of Theorem \ref{theorem1}, it is easy to obtain
the Probability of False Alarm as $\pfa=\Qfun_{F(M_1 \Nb,M_2 \Nb)}(\gamma)$. Then for fixed $\pfa$, $M_1$, $M_2$, and $\Nb$ the lower and upper bounds on the Probability of Detection, $P_{D,lb} \le \pdet \le P_{D,ub}$, can be respectively expressed as
$P_{D,lb}= \Qfun_{F(M_1 \Nb,M_2 \Nb)}({\tilde{\eta}_{lb}}\Qfun_{F(M_1 \Nb,M_2 \Nb)}^{-1}(\pfa))$ and
$P_{D,ub}= \Qfun_{F(M_1 \Nb,M_2 \Nb)}({\tilde{\eta}_{ub}}\Qfun_{F(M_1 \Nb,M_2 \Nb)}^{-1}(\pfa))$, where 
\beq
\tilde{\eta}_{lb}\triangleq\frac{\sigma_0^2+\sigma_x^2H \lambda_{1}(\mH^T \mPhi_o^T \mPhi_o \mH)}{\sigma_0^2+\sigma_x^2 \lambda_{M_1}(\mH^T \mPhi_s^T \mPhi_s \mH)},
\label{eq:generalLB}
\eeq
and 
\beq
\tilde{\eta}_{ub}\triangleq\frac{\sigma_0^2+\sigma_x^2\lambda_{M_2}(\mH^T \mPhi_o^T \mPhi_o \mH)}{\sigma_0^2+\sigma_x^2 \lambda_{1}(\mH^T \mPhi_s^T \mPhi_s \mH)}.
\label{eq:generalUB}
\eeq
Therefore the optimum design for $\mPhi_s$ and $\mPhi_o$ is the one that minimizes $\tilde{\eta}_{lb}$ and $\tilde{\eta}_{ub}$ in (\ref{eq:generalLB}) and (\ref{eq:generalUB}), which in turn is the one that minimizes $\lambda_{1}(\mH^T \mPhi_o^T \mPhi_o \mH)$ and $\lambda_{M_2}(\mH^T \mPhi_o^T \mPhi_o \mH)$ and maximizes  $\lambda_{1}(\mH^T \mPhi_s^T \mPhi_s \mH)$ and  $\lambda_{M_1}(\mH^T \mPhi_s^T \mPhi_s \mH)$. The proof is then easily concluded by employing Poincar\'{e} Separation Theorem \cite[Theorem P 10.4.2]{raorao}. 
\end{proof}
}

{\revcol
{\bf Remark 3}: The proposed detector can be thought of as an F-test detector. F-test \cite{degroot02} is a statistical test for comparing the variances of two normal populations. The proposed detector can be seen as an F-test detector where the two populations, i.e., outputs of the two sets of measurement devices, have been designed so as to impose maximum difference to their variances, which optimizes the performance of the F-detector since it works based on the difference between variances. This is in fact why the optimizations in Theorem \ref{thm0} and Corollary \ref{corol0}, which are trying to optimize different cost functions, result in identical solutions. \qed
}


{\revcol
While Theorem \ref{theorem1}  proposes a test statistic for signal detection and provides a proof for the performance of the proposed detector, it does not provide us with the optimum number of equipped samplers. Because of the complicated form of the tail probability of F-distribution and regularized incomplete Beta function, finding a closed form solution for $\pdet$ as a function of $\pfa$, $M_1$, $M_2$, and $\Nb$ is overly cumbersome. However, for large degrees of freedom, say larger than 100, it is possible to provide an analysis based on the following Normal approximation of the F-distribution \cite[Section 12.4.4]{krish06}
\beq
\Qfun_{F(d_1,d_2)}(x)\approx\Qfun_z(x-\tilde{\mu}/\tilde{\sigma})
\eeq
where $\Qfun_z$ is the tail function of the standard Gaussian distribution, 
\beq
\tilde{\mu}=\frac{d_2}{d_2-2},
\label{eq:mutilde}
\eeq
and
\beq
\tilde{\sigma}=\tilde{\mu} \sqrt{\frac{2(d_1+d_2-2)}{d_1(d_2-4)}}.
\label{eq:sigtilde}
\eeq
Then, for our problem, we can write
\beq
P_{D,lb} = \Qfun_z(\eta_{lb} \tilde{\sigma} \Qfun_z^{-1}(\pfa)+\eta_{lb} \tilde{\mu}),
\eeq
and similarly for $P_{D,ub}$, where $d_1\triangleq \alpha M \Nb$, $d_2 \triangleq (1-\alpha) M \Nb$, and $\alpha\triangleq M_1/M$. From (\ref{eq:mutilde}) and (\ref{eq:sigtilde}) it is clear that for big values of $M \Nb$, both $\tilde{\mu}$ and $\tilde{\sigma}$ are decreasing functions of $M \Nb$, and therefore $P_{D,lb}$ is an increasing function of both $M$ and $\Nb$.
This can be seen also from Figure \ref{fig:MNb} which illustrates an example of how the performance of the compressive detector changes with the total number of sampling devices $M$ and the number of temporal measurements $\Nb$. In this example, we set $\rho_1=\rho_{M_1}=1$ and $\rho_N=\rho_{N-M_2+1}=0$ to have $\pdet=P_{D,lb}=P_{D,ub}$. Besides, we choose $M_1=M_2=M/2$, $\pfa=0.05$, and $\sigma_x^2=\sigma_o^2=1$. The probabilities of false alarm and detection are then calculated from (\ref{eq:pfa}) and (\ref{eq:pdlb}).
\begin{figure}[t!]
\begin{center}
\includegraphics[width=0.7\textwidth,height=0.5\textwidth]{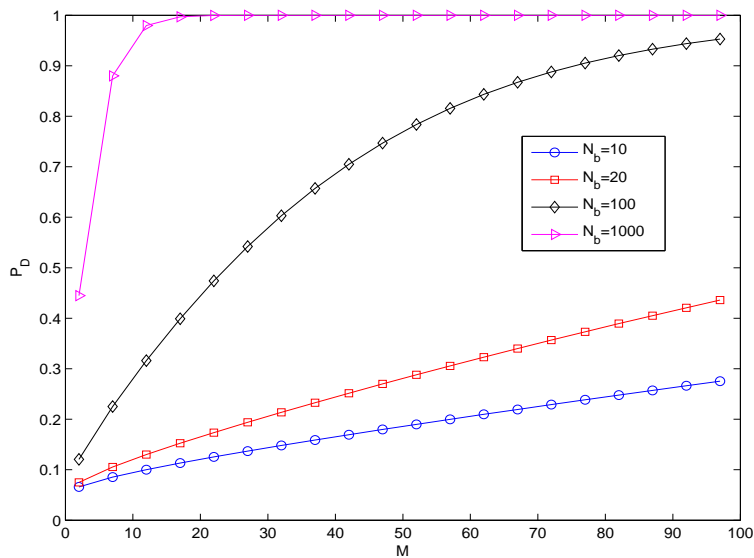}
\caption{{\revcol Probability of Detection as a function of $M$ and $\Nb$. Here we have $M_1=M_2=M/2$.}}
\label{fig:MNb}
\end{center}
\end{figure}
Notice that from a practical point of view an important factor which limits utilizing high values of $M_1$ and $M_2$ is the hardware cost and complexity as it is proportional to the number of samplers which is in turn determined by $M_1$ and $M_2$.

Furthermore, for $\Nb \gg 1$, one can apply a further approximation on (\ref{eq:mutilde}) and (\ref{eq:sigtilde}) to obtain
\beq
\tilde{\mu} \approx 1,
\label{eq:mutilde2}
\eeq
and
\beq
\tilde{\sigma} \approx \sqrt{\frac{2}{\alpha (1-\alpha) \Nb M}}.
\label{eq:sigtilde2}
\eeq
Equation (\ref{eq:sigtilde2}) then shows that the performance is optimized for $\alpha=0.5$, or in other words when we choose $M_1=M_2$. This means that when designing a Maximally-Uncorrelated Compressive Detector, the optimum choice is to divide the hardware budget equally between max-energy samplers $\mPhi_s$ and min-energy samplers $\mPhi_o$.
\begin{figure}[h]
\begin{center}
\includegraphics[width=0.7\textwidth,height=0.5\textwidth]{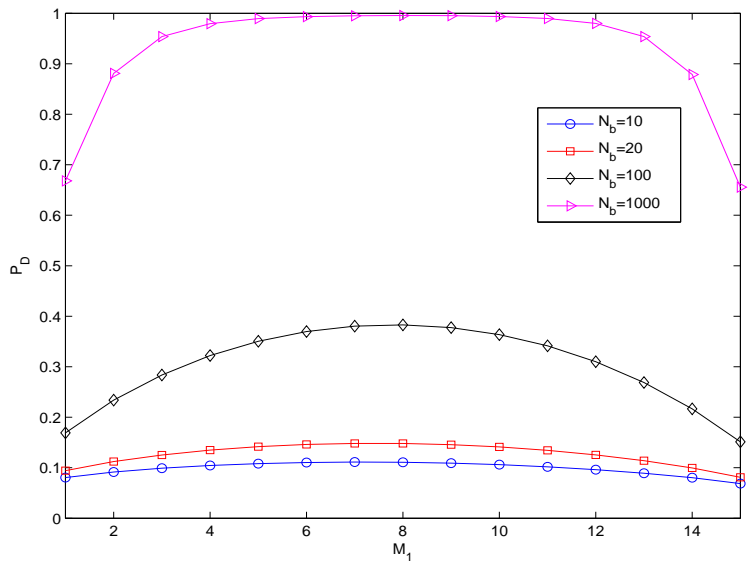}
\caption{{\revcol Probability of Detection as a function of $M_1$ and $\Nb$ for a fixed total number of measurement devices $M=16$.}}
\label{fig:M1}
\end{center}
\end{figure}
This can be seen also from Figure \ref{fig:M1} which shows that for a fixed number of measurement devices $M=16$, the maximum indeed occurs practically in the middle, which means that the optimum choice is to allocate half of the measurement devices to $\mPhi_s$ and the other half to $\mPhi_o$.}

\subsection{Fully-Correlated Compressive Detector}
\label{sec:maxsnr}

In this part, we study the effect of increasing the Signal-to-Noise Ratio (SNR) at the expense of uncorrelatedness of measurement devices. For a certain measurement device, the effective noise energy can be decreased (and therefore the SNR can be increased) by repeating the measurements and averaging over all of them. Collecting the multiple measurements one after the other over time for enhancing the effective SNR is not however a suitable solution here, firstly because the signal itself is random too, and secondly because  it increases the {\it latency} which is an issue in many applications such as radar or wireless communications due to the relatively quick changes of communication channel characteristics.

However in our proposed measurement framework, improving the effective SNR can be done by employing $M_1$ identical measurement vectors instead of $M_1 \times N$  measurement matrix $\mPhi_s$ in (\ref{eq:measS}) and $M_2$ identical measurement vectors instead of  $M_2 \times N$ measurement matrix $\mPhi_o$ in (\ref{eq:measO}) and then averaging over all measurements. This will not increase the latency because the identical measurements are taken in parallel. In other words, matrices $\mPhi_s$ and $\mPhi_o$ in (\ref{eq:measS}) and (\ref{eq:measO}) should be respectively redefined as
\beq
\mPhi_s \triangleq  \frac{1}{\sqrt{M}} \vone_{M_1,1} \otimes \vu_1^\tr,
\label{eq:measPrecS}
\eeq
and 
\beq
\mPhi_o \triangleq \frac{1}{\sqrt{M}} \vone_{M_2,1} \otimes \vu_{N}^\tr,
\label{eq:measPresO}
\eeq
where $\vu_1$ and $\vu_N$ are the left singular vectors of $\mH$ corresponding to the $\rho_1$ and $\rho_N$, respectively. The following theorem is to show the significance of the measurement design in (\ref{eq:measPrecS}) and (\ref{eq:measPresO}).

\begin{thm}
The measurement sub-matrix $\mPhi_s$ in (\ref{eq:measPrecS}) is the solution to the following optimization problem
\beq
\arg \max_{\mPhi} \Exp_{\vx,\vw}( \| \vz_s\|_2^2).
\label{eq:th00eq1}
\eeq
In other words, the rows of $\mPhi_s$ in (\ref{eq:measPrecS}) represent the set of $M_1$ measurement devices that maximize the total energy (or the total variance) in their outputs.

Correspondingly, the measurement sub-matrix $\mPhi_o$ in (\ref{eq:measPresO}) is the solution to the following optimization problem
\beq
\arg \min_{\mPhi}  \Exp_{\vx,\vw}( \| \vz_s\|_2^2).
\label{eq:th00eq2}
\eeq
In other words, the rows of $\mPhi_o$ in (\ref{eq:measPresO}) represent the set of $M_2$ measurement devices that minimize the total energy (or the total variance) in their outputs.
\label{thm00}
\end{thm}
\begin{proof}
We first notice that
\beqa
\Exp_{\vx,\vw}( \| \vz_s\|_2^2) &=&  \Exp_{\vx} (\vx^\tr \mH^\tr \mPhi_s^\tr \mPhi_s \mH \vx) + \frac{M_1 \sigma_0^2}{M} \nonr \\
&=& \sigma_x^2 \trace(\mH \mH^\tr \mPhi_s^\tr \mPhi_s) + \frac{M_1 \sigma_0^2}{M} 
\label{eq:th3proof1}
\eeqa
Since the second term on the right-hand side of (\ref{eq:th3proof1}) is independent of measurement design, the optimum design for $\mPhi_s$ is the one which maximizes the first term. Noticing that $\mH^\tr \mH$ and $\mPhi_s^\tr \mPhi_s$ are both symmetric, from \cite[P 11.4.5]{raorao} we have
\beqa
\trace(\mH \mH^\tr \mPhi_s \mPhi_s^\tr) \le \sum_{i=1}^{M_1} \rho_i^2 \lambda_i (\mPhi_s \mPhi_s^\tr)
\label{eq:th3proof2}
\eeqa
Unlike Theorem \ref{thm0}, here we have now no constraint on matrix $\mPhi_s$ other than its rows having norm $1/\sqrt{M}$ which means that $\trace(\mPhi_s \mPhi_s^\tr)=\sum_{i=1}^{M_1} \lambda_i(\mPhi_s \mPhi_s^\tr)=M_1/M$. Therefore the design of $\mPhi_s$ can be done in two steps. In the first step we find the eigenvalues $\{ \lambda_i(\mPhi_s \mPhi_s^\tr)\}_{i=1}^{M_1}$ that maximize the right-hand side of inequality (\ref{eq:th3proof1}). From \cite[4.B.7 and 1.A]{majorization} the maximum is attained when $\lambda_1(\mPhi_s \mPhi_s^\tr)=1$ and $\{ \lambda_i(\mPhi_s \mPhi_s^\tr)\}_{i=2}^{M_1}=0$. In the second step we simply employ \cite[P 11.4.5]{raorao} to conclude the proof of (\ref{eq:measPrecS}) being the solution to (\ref{eq:th00eq1}). The proof for  (\ref{eq:measPresO}) is similar.
\end{proof}

After designing the samplers as in (\ref{eq:measPrecS}) and (\ref{eq:measPresO}) and collecting measurement vectors $\vz_s[n]$ and $\vz_o[n]$, the effective energy of noise can then be reduced by averaging over all elements of $\vz_s[n]$ and $\vz_o[n]$ to obtain
\beq
\bar{z}_s[n]=\frac{1}{M_1} \vone_{1,M_1} \vz_s[n],~n=1,\ldots,\Nb,
\label{eq:avgMeasS}
\eeq
and
\beq
\bar{z}_o[n]=\frac{1}{M_2} \vone_{1,M_2} \vz_o[n],~n=1,\ldots,\Nb.
\label{eq:avgMeasO}
\eeq
In other words, instead of choosing $M_1$ rows of $\mPhi_s$ that span the $M_1$-dimensional subspace along which the energy of signal is maximum as in Section \ref{sec:maxrank}, we choose $M_1$ identical rows for $\mPhi_s$ each as the $1$-dimensional subspace along which the energy of signal is maximum, and then decrease the noise power (and therefore increase the effective SNR) by averaging over all $M_1$ measurements. Similarly, for $\mPhi_o$, instead of choosing its $M_2$ rows that span the $M_2$-dimensional subspace along which the energy of signal is minimum as in Section \ref{sec:maxrank}, we choose $M_2$ identical rows for $\mPhi_o$ each as the $1$-dimensional subspace along which the energy of signal is minimum, and thus decrease the noise energy (and therefore increase the effective SNR) by averaging over all $M_2$ measurements.

The following theorem introduces the test for performing (\ref{eq:test1}) based on measurement design in (\ref{eq:measPrecS}) and (\ref{eq:measPresO}).
\begin{thm}
\label{thm:fccd}
Consider the following test for compressive detection of subspace signal $\mH \vx$ 
\beq
\cT = \frac{M_1 \sum\limits_{n=1}^{\Nb} |\bar{z}_s[n]|^2}{M_2 \sum\limits_{n=1}^{\Nb} |\bar{z}_o[n]|^2}\overset{\cH_1}{\underset{\cH_0}{\gtrless}} \gamma
\label{eq:testStatPrec}
\eeq
The probability of false alarm $\pfa$ and the probability of detection $\pdet$ are then expressed as
\beq
\pfa=\Qfun_{F(\Nb,\Nb)}(\gamma),
\label{eq:pfaPrec}
\eeq
and
\beq
\pdet=\Qfun_{F(\Nb,\Nb)}\Big(\frac{\sigma_0^2+M_2 \sigma_x^2 \rho_N^2}{\sigma_0^2+M_1 \sigma_x^2 \rho_1^2}\gamma\Big).
\label{eq:pdPrec}
\eeq
\end{thm}
\begin{proof}
To derive the probability of false alarm $\pfa$ and the probability of detection $\pdet$, we first notice that  if the Gaussian noise is i.i.d, then it is easy to observe that the distribution of $\bar{z}_o[n]$ is
\beq
\bar{z}_o[n] \sim \left\{ \begin{array}{ll}  \cN(0,\frac{\sigma_0^2}{M M_2}) & \mathrm{under~} \cH_0,  \\ \\
 \cN(0,\frac{\sigma_0^2}{M M_2}+\frac{1}{M}\sigma_x^2 \rho_N^2)& \mathrm{under~} \cH_1. \end{array}\right.
\label{eq:vzs2}
\eeq
Similarly, the distribution of $\bar{z}_s[n]$ is
\beq
\bar{z}_s[n] \sim \left\{ \begin{array}{ll}  \cN(0,\frac{\sigma_0^2}{M M_1}) & \mathrm{under~} \cH_0,  \\ \\
 \cN(0,\frac{\sigma_0^2}{M M_1}+\frac{1}{M}\sigma_x^2 \rho_1^2)& \mathrm{under~} \cH_1, \end{array}\right.
\label{eq:vzs3}
\eeq
Then by following steps similar to those in Theorem \ref{theorem1}, (\ref{eq:pfaPrec}) and (\ref{eq:pdPrec}) are concluded.
\end{proof}
{\revcol
The Fully-Correlated Compressive Detector of (\ref{eq:measPrecS}) and (\ref{eq:measPresO}) whose probabilities of false alarm and detection were derived in Theorem \ref{thm:fccd} is optimum in terms of providing the largest difference between the expected energy of its two set of samplers. However as discussed earlier, a major measure of interest in statistical hypothesis testing is the power of the test.
The following corollary discusses the optimality of the proposed Fully-Correlated Compressive design and detector in terms of the power of the test.

\begin{cor}
\label{corol2}
Among all fully-correlated measurement designs of the form
\beq
\mPhi_s \triangleq  \frac{1}{\sqrt{M}} \vone_{M_1,1} \otimes \vt_1^\tr,
\label{eq:measPrecS_gen}
\eeq
and 
\beq
\mPhi_o \triangleq \frac{1}{\sqrt{M}} \vone_{M_2,1} \otimes \vt_{2}^\tr,
\label{eq:measPresO_gen}
\eeq
with $\|\vt_1\|_2=\|\vt_2\|_2=1$ and $\vt_1^\tr \vt_2=0$, the one proposed in (\ref{eq:measPrecS}) and (\ref{eq:measPresO}) delivers the highest $\pdet$ for a fixed $\pfa$ for performing the test given in (\ref{eq:testStatPrec}).
\begin{proof}
By taking steps similar to the proof of Theorem \ref{thm:fccd}, the probability of false alarm can be derived as $\pfa=\Qfun_{F(\Nb,\Nb)}(\gamma)$. Therefore, for a given $\pfa$, $M_1$, $M_2$, and $\Nb$, we can write
\beq
\pdet=\Qfun_{F(\Nb,\Nb)}\Big(\frac{\sigma_0^2+M_2 \sigma_x^2 \vt_2^T \mH \mH^T \vt_2 }{\sigma_0^2+M_1 \sigma_x^2  \vt_1^T \mH \mH^T \vt_1}\Qfun_{F(\Nb,\Nb)}^{-1}(\pfa)\Big).
\label{eq:pdPrecGen}
\eeq
From (\ref{eq:pdPrecGen}), it is clear that the most powerful design is the one that minimizes $\frac{\sigma_0^2+M_2 \sigma_x^2 \vt_2^T \mH \mH^T \vt_2 }{\sigma_0^2+M_1 \sigma_x^2  \vt_1^T \mH \mH^T \vt_1}$, which in turn is the one that minimizes $ \vt_2^T \mH \mH^T \vt_2$ and maximizes $ \vt_1^T \mH \mH^T \vt_1$. Then by employing the Rayleigh-Ritz theorem \cite[Chapter 8]{zhang2011matrix} it can be easily proven that the optimum design is $\vt_1=\vu_1$ and $\vt_2=\vu_N$.
\end{proof}
\end{cor}
Again, similar to the Maximally-Uncorrelated case, Theorem \ref{thm00} and Corollary \ref{corol2} result in identical solution although they are trying to optimize different cost functions, which is because the proposed detector is in fact an $F$-test detector working based on the difference between the variances of the two sets of measurements.
}

The performance of Fully Correlated Detector has a clear connection to the way we choose the number of samplers. The following Corollary discusses this.

\begin{cor}
\label{cor1}
For a given number of samplers $M$ and probability of false alarm $\pfa$, the performance of the detector $\pdet$ is maximized when $M_2=1$ and $M_1=M-1$.
\end{cor}
\begin{proof}
First, we notice from (\ref{eq:pfaPrec}) that the threshold level $\gamma$ is independent of the choice of $M_1$ and $M_2$. Then from (\ref{eq:pdPrec}) it is clear that $\pdet$ achieves its maximum value when $M_2$ reaches its minimum value and $M_1$ reaches its maximum value. 
\end{proof}

The above Corollary asserts that in the case of Fully-Correlated Detector, for a fixed hardware budget, in terms of the total number of  samplers $M=M_1+M_2$, the best performance, in terms of $\pdet$, is obtained if we employ just one sampler for $\mPhi_o$ and spend the rest of the hardware budget on $\mPhi_s$. Besides, by putting the optimum values $M_2=1$ and $M_1=M-1$ in (\ref{eq:pdPrec}) it is clear that the performance enhances as the total number of measurement devices $M$ grows.
\begin{figure}[t!]
\begin{center}
\includegraphics[width=0.7\textwidth,height=0.5\textwidth]{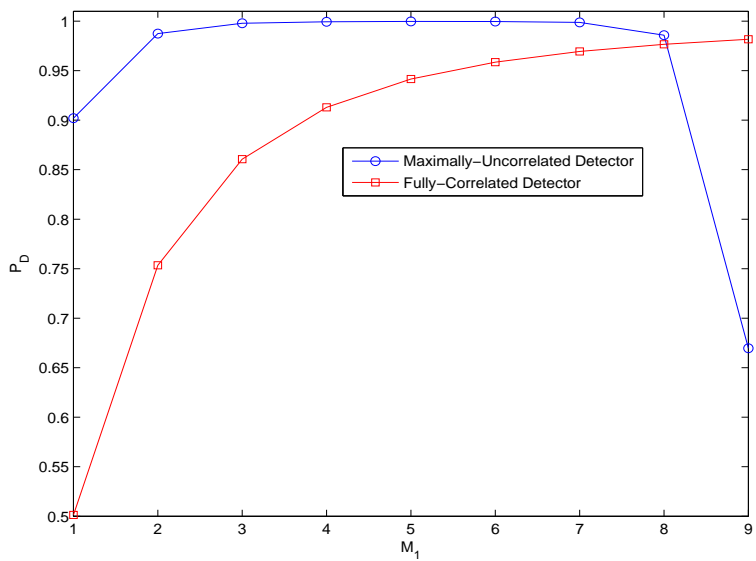}
\caption{\revcoltwo Probability of Detection as a function of $M_1$ for the two detectors introduced in Sections \ref{sec:maxrank} and \ref{sec:maxsnr}. The total number of samplers is $M=10$, the Signal-to-Noise Ratio is $\frac{\sigma_x^2}{\sigma_0^2}=0 \mathrm{dB}$ and the largest and smallest eigenvalues of $\mH \mH^\tr$ are $\rho_1^2=10$ and $\rho_N^2$=0, respectively. The Probability of False Alarm is $\pfa=0.01$ and the number of temporal measurements is $\Nb=5$.}
\label{fig:PdVsM1}
\end{center}
\end{figure}
Figure \ref{fig:PdVsM1} demonstrates an example of how the probability of detection for Fully-Correlated detector changes with the number of maximum energy samplers $M_1$. The performance curve of the Maximally-Uncorrelated case is also shown for comparison. As it can be observed, while the performance of Fully-Correlated detector clearly enhances with increasing $M_1$, the performance of Maximally-Uncorrelated Detector does not show such a monotonic relationship. The best performance of the Maximally-Uncorrelated Detector here is achieved for $M_1=M_2=5$. The probabilities of detection of the two detectors here have been calculated from (\ref{eq:pdPrec}) and (\ref{eq:pdlb}). The parameters of the Maximally-Uncorrelated detector has been chosen such that $\pdet=P_{D,lb}=P_{D,ub}$.

{\bf Remark 4}: Comparing (\ref{eq:th00eq1}) and (\ref{eq:th00eq2}) to (\ref{eq:th0eq1}) and (\ref{eq:th0eq2}) it may seem that the performance of the Fully-Correlated detector is always superior to the Maximally-Uncorrelated detector as the optimizations in (\ref{eq:th00eq1}) and (\ref{eq:th00eq2}) are unconstrained compared to the optimizations in (\ref{eq:th0eq1}) and (\ref{eq:th0eq2}) which are constrained. However, we should notice that this is not always the case as Maximally-Uncorrelated detector provides a higher {\it diversity} in its collected samples, which is reflected by the higher {degrees of freedom} in (\ref{eq:pfa}), (\ref{eq:pdlb}) and (\ref{eq:pdub}) compared to (\ref{eq:pfaPrec}) and (\ref{eq:pdPrec}). For example consider the case when $\vx[n]$ is present but lies orthogonal to the strongest right singular vector of  $\mH$. Then for the Fully-Correlated samplers the energy of samples collected by all of the $M_1$ samplers $\mPhi_s$ includes only the energy of noise while for the Maximally-Uncorrelated Compressive Detector only the output energy of the first max-energy sampler is equivalent to the noise energy and the output of the rest of $M_1-1$ samplers include both the signal and the noise energy.
{\revcoltwo This can also be seen from the example of Figure \ref{fig:PdVsM1} where the best performance of the Fully-Correlated detector occurring at $M_1=M-1=9$ is inferior to the best performance of Maximally-Uncorrelated detector occurring at $M_1=M/2=5$.}
\qed

\subsection{Special Case of Known Noise Variance}
\label{sec:known}
As discussed earlier in this Section, the idea of having two sets of measurement devices $\mPhi_s$ and $\mPhi_o$ is to have two sets of measurements with maximum difference in their variances when the signal truly exists and equivalent variances when the signal does not exist and then making decision about the existence of signal based on this difference. However, if the variance of noise is known, we can simplify the proposed design by discarding $\mPhi_o$ as we can simply compare the variance of the output of $\mPhi_s$ to the known noise variance. In other words, we can employ only $M_1$ measurement devices determined by the rows of $\mPhi_s$ and then compare the power of measurements taken by them to the power of noise: if they are almost the same then we decide that no signal exists, if not then we decide that signal exists.

Similar to the two detectors proposed in sections \ref{sec:maxrank} and \ref{sec:maxsnr}, we can have two detectors when the noise variance is known. The following two propositions express these two detectors and their performances. We remark that since $\mPhi_o$ is discarded when noise variance is known, we will have $M_2=0$ and therefore the whole number of measurement devices is $M=M_1$.

\begin{prop} 
\label{lem1}
If we design the measurement matrix according to (\ref{eq:measS}), then in the case of known $\sigma_0^2$, the test in (\ref{eq:testStat}) can be replaced by 
\beq
\cT = \frac{\sum\limits_{n=1}^{\Nb} \vz_s[n]^\tr \vz_s[n]}{\sigma_0^2/M_1}\overset{\cH_1}{\underset{\cH_0}{\gtrless}} \gamma.
\label{eq:testStatKnown}
\eeq
The Probability of False Alarm, $\pfa$, for this test is computed as
\beq
\pfa=\Qfun_{\chi^2 (M_1 \Nb)}(\gamma),
\label{eq:pfaKnown}
\eeq
and the Probability of Detection, $\pdet$, is bounded as
\beq
P_{D,lb} \le \pdet \le P_{D,ub}
\label{eq:pdKnown}
\eeq
where $P_{D,lb}$ and $P_{D,ub}$ are, respectively, the lower bound and the upper bound for the probability of detection which will take the following formulas
\beq
P_{D,lb}= \Qfun_{\chi^2(M_1 \Nb)}({\eta_{lb}^\prime}\gamma),
\label{eq:pdlbKnown}
\eeq
and
\beq
P_{D,ub}= \Qfun_{\chi^2 (M_1 \Nb)}({\eta_{ub}^\prime}\gamma),
\label{eq:pdubKnown}
\eeq
where $\eta_{lb}^\prime \triangleq\frac{\sigma_0^2}{\sigma_0^2+\sigma_x^2 \rho_{M_1}^2}$, $\eta_{ub}^\prime\triangleq\frac{\sigma_0^2}{\sigma_0^2+\sigma_x^2 \rho_{1}^2}$, and $\Qfun_{\chi^2 (d)}(x)$ is the {\it tail probability} of the chi-squared distribution with parameter (degree of freedom) $d$ at point $x$.
\end{prop}

\begin{prop}
\label{lem2}
If we design the measurement matrix according to (\ref{eq:measPrecS}), then in the case of known $\sigma_0^2$, the test in (\ref{eq:testStatPrec}) can be replaced by 
\beq
\cT = \frac{\sum\limits_{n=1}^{\Nb} |\bar{z}_s[n]|^2}{\sigma_0^2/M_1^2}\overset{\cH_1}{\underset{\cH_0}{\gtrless}} \gamma
\label{eq:testStatPrecKnown}
\eeq
The Probability of False Alarm, $\pfa$, for this test is 
\beq
\pfa=\Qfun_{\chi^2(\Nb)}(\gamma),
\label{eq:pfaPrecKnown}
\eeq
and
\beq
\pdet=\Qfun_{\chi^2(\Nb)}\Big(\frac{\sigma_0^2}{\sigma_0^2+M_1 \sigma_x^2 \rho_1^2}\gamma\Big).
\label{eq:pdPrecKnown}
\eeq
\end{prop}

{\revcoltwo
\subsection{Compressive Detection in Presence of Known Subspace Interference}
\label{sec:interference}

Extending the above results to the subspace interference case is straightforward and is similar to the same scenario for classical uncompressive detection of deterministic subspace signals in presence of interference in \cite{scharf94}.
Suppose that the subspace signal model in (\ref{eq:sysModel}) is replaced by
\beq
\vs[n] = \mH \vx[n] + \mG \vt[n],
\label{eq:sysModelInterference}
\eeq
where $\mG$ is an $N \times K^\prime$ matrix and the second term on the right-hand side of (\ref{eq:sysModelInterference}) represents the subspace interference. If $\vs[n]$ were available, the interference could be cancelled by left multiplying (\ref{eq:sysModelInterference}) by matrix $\mP_\mG^{\perp}=\mI - \mG (\mG^\tr \mG)^{-1} \mG^\tr$ to obtain
\beq
\mP_\mG^{\perp} \vs[n] = \mP_\mG^{\perp} \mH \vx[n].
\label{eq:sysModelInterference2}
\eeq
Comparing (\ref{eq:sysModelInterference2}) to (\ref{eq:sysModel}), it can be easily deduced that in this scenario to design the measurement vectors, i.e., rows of $\mPhi$, the left-singular vectors of matrix $\mH$ in the interference-free scenario in (\ref{eq:measS}) and (\ref{eq:measO}) must be replaced by  left-singular vectors of matrix $\mP_\mG^{\perp} \mH$ right-multiplied by matrix $\mP_\mG^{\perp}$.} 
Otherwise, all the detector developments are similar to the previous subsections.

\section{Compressive Detector with Non-Ideal Measurements}
\label{sec:nonideal}

In this section we investigate the case when due to the lack of complete knowledge about the subspace or due to the limitations in designing hardware with infinite precision, the measurement vectors do not exactly match the singular vectors of matrix $\mH$. We consider two models for describing this uncertainty. The first model, defines the finite precision using additional observation noise. The second model, on the other hand, chooses the measurement vector from a set of realizable measurement vectors (hardware) which are approximations (or quantized versions) of infinite-precision real-valued singular vectors. 

\subsection{Compressive Detection with Imprecise Measurements}
\label{subsec:imprecise}
{\revcoltwo
In this section we analyze the detector when the measurement device is not precise. This imprecision in measurements can be modeled as \cite{distilled}
\beq
\tilde{\vz}[n]=\vz[n] + \delta^{-1/2} \veps[n],
\eeq
where $\tilde{\vz}[n]$ denotes the imprecise compressive measurements, $\vz[n]$ denotes the basic measurements as introduced in Section \ref{sec:sysmodel}, $\veps[n]\sim\cN(\vzero_{M,1},\mI_{M})$, and $\delta$ represents the precision of measurements. Thus the imprecision is here modeled through additional observation or measurement noise, as in \cite{distilled}. 

Now, if we replace $\vz_s[n]$ and $\vz_o[n]$ in the test statistic (\ref{eq:testStat}) of Maximally-Uncorrelated detector by $\tilde{\vz}_s[n]\triangleq\tilde{\vz}[n]\Big|_{1:M_1}$ and $\tilde{\vz}_o[n]\triangleq\tilde{\vz}[n]\Big|_{M_1+1:M}$, }respectively, then it can be easily verified that the probability of false alarm will be the same as in (\ref{eq:pfa}), but the lower and upper bounds of probability of detection in (\ref{eq:pdlb}) and (\ref{eq:pdub}) will be replaced by 
\beq
P_{D,lb}=\Qfun_{F(M_1 \Nb,M_2 \Nb)}({\frac{\sigma_0^2+\sigma_x^2 \rho_{N-M_2+1}^2 + M \delta^{-1}}{\sigma_0^2+\sigma_x^2 \rho_{M_1}^2 + M \delta^{-1}}\gamma}),
\label{eq:Precpdlb}
\eeq 
and 
\beq 
P_{D,ub}=\Qfun_{F(M_1 \Nb,M_2 \Nb)}({\frac{\sigma_0^2+\sigma_x^2 \rho_{N}^2 + M \delta^{-1}}{\sigma_0^2+\sigma_x^2 \rho_{1}^2 + M \delta^{-1}}\gamma}),
\label{eq:Precpdub}
\eeq
respectively.
Comparing (\ref{eq:Precpdlb}) and (\ref{eq:Precpdub}) with (\ref{eq:pdlb}) and (\ref{eq:pdub}), it is clear that the performance degrades as the precision of measurement devices $\delta$ decreases. This can be however compensated by improving the precision through employing several, say $\nk$, identical measurement devices instead of each device and then averaging over all measurements. In other words, instead of $M$ measurement devices, we employ $M \nk$ to improve the performance. 
We call $\nk$ the {\it hardware budget factor} as it indicates the number of measurement devices that we employ.

In this case, the probability of false alarm will be still expressed as in (\ref{eq:pfa}) while the lower and upper bounds for probability of detection are modified as
\beq
P_{D,lb}=\Qfun_{F(M_1 \Nb,M_2 \Nb)}({\frac{\sigma_0^2+\nk \sigma_x^2 \rho_{N-M_2+1}^2 + M \delta^{-1}}{\sigma_0^2+\nk \sigma_x^2 \rho_{M_1}^2 + M \delta^{-1}}\gamma}),
\label{eq:PrecpdlbK}
\eeq 
and 
\beq 
P_{D,ub}=\Qfun_{F(M_1 \Nb,M_2 \Nb)}({\frac{\sigma_0^2+\nk \sigma_x^2 \rho_{N}^2 + M \delta^{-1}}{\sigma_0^2+\nk \sigma_x^2 \rho_{1}^2 + M \delta^{-1}}\gamma}),
\label{eq:PrecpdubK}
\eeq 
Comparing (\ref{eq:PrecpdlbK}) and (\ref{eq:PrecpdubK}) with (\ref{eq:pdlb}) and (\ref{eq:pdub}), the hardware budget factor $\nk$ for achieving a performance similar or better than the infinite precision case of $\delta^{-1}=0$ is obtained as
\beq
\nk \ge 1+\frac{M}{\delta \sigma_0^2}.
\label{eq:L}
\eeq

\subsection{Compressive Detection Using a Fixed Set of Measurement Vectors}
\label{subsec:fixed}
The second model in this section for modeling non-ideal measurement devices is formulated as choosing the measurement vectors from a fixed set of unit-norm orthogonal vectors whose members are not necessarily singular vectors of the system matrix $\mH$. {These vectors can be hardware-realizable approximation of the singular vectors, e.g. quantized versions of them, or simple down-samplers.} Assume that $M=M_1+M_2$ measurement vectors should be chosen out of $R \ge M$ uncorrelated columns of $N \times R$ matrix $\mPsi=[\vfi_1,\vfi_2,\ldots,\vfi_R]$ (notice that $M \mPsi^\tr \mPsi=\mI_R$). The rows of $\mPhi_s$ and $\mPhi_o$ then have to be chosen from the columns of $\mPsi$. 


The optimum $\mPhi_s$ is the one that maximizes the energy of its output. If we want to follow the Maximally-Uncorrelated design of Section \ref{sec:maxrank}, then the rows of $\mPhi_s$ are the $M_1$ columns of $\mPsi$ for which $\|\mPhi_s \mH\|_F$ is maximized, i.e.,
\beq
\mPhi_s = \underset{\mPsi_{\cM}^\tr}{\arg} \max_{\cM\subset\cR:|\cM|=M_1} \trace (\mH \mH^\tr \mPsi_{\cM} \mPsi_{\cM}^\tr),
\eeq
where $\cR=\{1,2,\ldots,R\}$ and $\mPsi_{\cM}$ is the sub-matrix of $\mPsi$ consisting of columns determined by index set $\cM$.
Similarly, the rows of $\mPhi_o$ can be selected as the $M_2$ columns of $\mPsi$ for which $\|\mPhi_s \mH\|_F$ is minimized, i.e.,
\beq
\mPhi_o =  \underset{\mPsi_{\cM}^\tr}{\arg} \min_{\cM\subset\cR:|\cM|=M_2} \trace (\mH \mH^\tr \mPsi_{\cM} \mPsi_{\cM}^\tr)
\eeq

On the other hand, if we want to follow the Fully-Correlated design of Section \ref{sec:maxsnr}, then we should pick the two columns of $\mPsi$, say $\vfi_i$ and $\vfi_j$, that, respectively, maximize and minimize the signal energy $\sigma_x^2 \vfi^\tr \mH \mH^T \vfi$ and then design $\mPhi_s$ and $\mPhi_o$, respectively, as

\beq
\mPhi_s \triangleq  \frac{1}{\sqrt{M}} \vone_{M_1,1} \otimes \vfi_i^\tr,
\label{eq:measPrecSSpecific}
\eeq
and 
\beq
\mPhi_o \triangleq \frac{1}{\sqrt{M}} \vone_{M_2,1} \otimes \vfi_{j}^\tr.
\label{eq:measPresOSpecific}
\eeq

To see the performance of the detectors in this case, we first provide the detector's formulation when the noise variance is known. It is easy to verify that when the elements of $\cM$ are orthonormal, $\pfa$ is computed as in (\ref{eq:pfaKnown}) and (\ref{eq:pfaPrecKnown}) for the Maximally-Uncorrelated and Fully-Correlated detectors, respectively. However, $\pdet$ of Fully-Correlated detector as well as the bounds of $\pdet$ of Maximally-Uncorrelated detector will differ from what we have in (\ref{eq:pdlbKnown}), (\ref{eq:pdubKnown}) and (\ref{eq:pdPrecKnown}). The reason for this is that the rows of $\mPhi_s$ are not any longer the singular vectors of $\mH$ and therefore the set of singular values of $\mPhi_s \mH$ is not a subset of the set of singular values of $\mH$. In other words, (\ref{eq:pdlbKnown}), (\ref{eq:pdubKnown}) and (\ref{eq:pdPrecKnown}) should be respectively replaced by the following equations 

\beq
P_{D,lb}= \Qfun_{\chi^2(M_1 \Nb)}({\tilde{\eta}_{lb}}\gamma),
\label{eq:pdlbKnownNI}
\eeq
and
\beq
P_{D,ub}= \Qfun_{\chi^2 (M_1 \Nb)}({\tilde{\eta}_{ub}}\gamma),
\label{eq:pdubKnownNI}
\eeq
and
\beq
\pdet=\Qfun_{\chi^2(\Nb)}\Big(\frac{\sigma_0^2}{\sigma_0^2+M_1 \sigma_x^2 \tilde{\rho}_1^2}\gamma\Big).
\label{eq:pdPrecKnownNI}
\eeq
where ${\tilde{\eta}}_{lb} \triangleq\frac{\sigma_0^2}{\sigma_0^2+\sigma_x^2 {\tilde{\rho}}_{M_1}^{2}}$, 
$\tilde{\eta}_{ub}\triangleq\frac{\sigma_0^2}{\sigma_0^2+\sigma_x^2 \tilde{\rho}_{1}^2}$, and $\tilde{\rho}_{1}$ and $\tilde{\rho}_{M_1}$ are the largest and smallest singular values of $\mPhi_s \mH$, respectively. Since the eigenvalues of $\mH \mH^\tr \mPhi_s^\tr \mPhi_s$ are smaller than the corresponding eigenvalues of $\mH \mH^\tr$ \cite[Theorem 6.76]{seber08}, the performance is lower than those indicated in (\ref{eq:pdlbKnown}), (\ref{eq:pdubKnown}) and (\ref{eq:pdPrecKnown}). 

When the variance of noise is unknown, on the other hand, we cannot write the probability of detection of the detector based on $F$-distribution as the outputs of $\mPhi_s$ and $\mPhi_o$ measurement devices are not independent of each other, which is in turn because the rows of $\mPhi_s$ and $\mPhi_o$ are not any longer  the singular vectors of $\mH$. However, we can still design a constant false alarm rate detector since in the absence of signal term $\mH \vx$ the output of these two measurement devices are still independent as $\mPhi_s \mPhi_o^\tr=\mZero_{M_1,M_2}$.



\section{Compressive Detection on a Finite Union of Subspaces}
\label{sec:union}

In this section, we generalize the results of previous sections to the case when we know that the $\mH \in \sS$ where
\beq
\sS = \{\mH_1,\mH_2,\ldots,\mH_Q\},
\label{eq:subspacesSet}
\eeq
where $\mH_q,~q=1,\ldots,Q,$ is a full-column rank $N \times K_i$ matrix and the probability that the signal lies in subspace $\mH_q$ is $\bP(\mH=\mH_q)=\pr_q,~q=1,\ldots,Q$. The goal is then to accomplish the test of (\ref{eq:test1}).

Similar to the case of $Q=1$, the measurement matrix is tailored to performing the $F$-test. In other words, it must consist of two parts: one lying in the strongest signal subspace, which attains the maximum energy, and the other in the weakest subspace, which attains the minimum energy. But, since the true $q$ is not known a priori, we should design it such that the expected value of the energy over the set $\sS$ gets its maximum value for the first part and its minimum value for the second part.

To be more precise, if we adopt the Maximally-Uncorrelated strategy, the first and the second sets of measurement devices will be chosen as the solutions of the optimization problems
\beqa
\arg \max_{\mPhi_s}\Exp_{\mH,\vx,\vw}(\vz_s^\tr \vz_s), \mathrm{~s.t.~} M \mPhi_s \mPhi_s^\tr = \mI_{M_1},
\label{eq:maxMulti}
\eeqa
and
\beqa
\arg \min_{\mPhi_o}\Exp_{\mH,\vx,\vw}(\vz_o^\tr \vz_o), \mathrm{~s.t.~} M \mPhi_o \mPhi_o^\tr = \mI_{M_2},
\label{eq:minMulti}
\eeqa
respectively. On the other hand if we want to design it based on Fully-Correlated strategy we have to relax the constraints $M \mPhi_s \mPhi_s^\tr = \mI_{M_1}$ and $M \mPhi_o \mPhi_o^\tr = \mI_{M_2}$ in (\ref{eq:maxMulti}) and (\ref{eq:minMulti}).
This will guarantee that the energy of the signals collected at the first and second channels attain, respectively, the maximum and minimum possible values which can then be exploited for optimum compressive signal detection.
\begin{thm}
The first and second channels of measurement matrix which satisfy (\ref{eq:maxMulti}) and (\ref{eq:minMulti}) are respectively the strongest $M_1$ eigenvectors and the weakest $M_2$ eigenvectors of the following matrix
\beq
\mH_{\mathrm{eq}} = \sum\limits_{q=1}^{Q} \pr_q \mH_q \mH_q^\tr,
\label{eq:Heq}
\eeq
\end{thm}
\begin{proof}
It is easy to see that
\beqa
\Exp_{\mH,\vx,\vw}(\vz_s^\tr \vz_s) &=& \Exp_{\mH} \Big( \sigma_x^2 \trace(\mH \mH^\tr \mPhi_s^\tr \mPhi_s) + \frac{\sigma_0^2}{M} \Big) \nonr \\
&=& \sum\limits_{q=1}^Q  \pi_q \sigma_x^2 \trace(\mH_q \mH_q^\tr \mPhi_s^\tr \mPhi_s) + \frac{\sigma_0^2}{M}. \nonr \\
&=& \sigma_x^2 \trace(\mH_{\mathrm{eq}} \mPhi_s^\tr \mPhi_s) + \frac{\sigma_0^2}{M}. \nonr \\
\label{eq:multiHproof}
\eeqa
Since the second term in the right-hand side of (\ref{eq:multiHproof}) is independent of the choice of the measurement matrix, the measurement sub-matrix $\mPhi_s$ that maximizes (\ref{eq:multiHproof}) is the one maximizing the first term, which contains the strongest $M_1$ eigenvectors of $\mH_{\mathrm{eq}}$ as its rows \cite[P 11.4.5]{raorao}. The proof of the weakest $M_2$ eigenvectors of $\mH_{\mathrm{eq}}$ being the solution of (\ref{eq:minMulti}) is similar.
\end{proof}
{\revcoltwo
{\bf Connection to the compressive detection of sparse signals}: The ordinary compressive detection of random sparse signals can be specified as a particular case of the setting in this section. To see this, assume that we have a $K$-sparse signal in $N$-dimensional basis $\Omegab$ as
\beq
\vs = \Omegab \vxx,
\label{eq:finiteUnionToCS}
\eeq
where $\vs \in \bR^{N}$ and $\vxx \in \bR^{N}$ are both vectors of size $N$ and only $K \ll N$ entries of $\vxx$ are nonzero. We assume that the $K$ nonzero entries of $\vxx$ are i.i.d with Gaussian distribution. Denoting the true unknown support of $\vxx$ by $\sI$, the signal model in (\ref{eq:finiteUnionToCS}) is equivalent to 
\beq
\vs = \Omegab_{\sI} \vx,
\label{eq:finiteUnionToCS2}
\eeq
where $\vx \in \bR^{K}$ and $\Omegab_{\sI} \in \bR^{N \times K}$ contains the columns of $\Omegab$ corresponding to index set $\sI$. In the absence of any further structure or knowledge of $\sI$, any $K$ entries of $\vxx$ can be its true support set equally likely with probability ${1}/{{N \choose K}}$. In this case, the members of the union set $\sS$ in (\ref{eq:subspacesSet}) are the $Q={N \choose K}$ size-$K$ combinations of the $N$ columns of $\Omegab$. Furthermore, we will have $\pi_1=\ldots=\pi_Q=\frac{1}{Q}$.

If the basis $\Omegab$ is orthonormal, then it is easy to verify from (\ref{eq:Heq}) that
\beq
\mH_{\mathrm{eq}} = \frac{K}{N} \mI_N.
\label{eq:HeqIn}
\eeq
Since all the singular values of $\mH_{\mathrm{eq}}$ in (\ref{eq:HeqIn}) are equal, then $\mH_{\mathrm{eq}}$ has no {\it structure} and therefore the methods proposed in Sections \ref{sec:maxrank} and \ref{sec:maxsnr} are not suitable for detection of this signal\footnote{Notice that in the absence of noise variance knowledge, none of the methods in the literature can detect the signal.} but we can use the method of Section \ref{sec:known} provided that the noise variance is known. However, if there is a known structure in the sparse random signal model of (\ref{eq:finiteUnionToCS}), for instance if the probability of different $K$-size combinations are not equal or if $\Omegab$ is not an orthonormal basis, in the sense that not all of the singular values of $\mH_{\mathrm{eq}}$ are identical, then the methods proposed in Sections \ref{sec:maxrank} and \ref{sec:maxsnr} can be employed for detection of signal from compressive measurements.}


\section{Simulation Results}
\label{sec:simulation}
In this section we study the performance of the proposed compressive detection methods via computer simulation studies. 

The first experiment is to compare the performance of the two methods proposed in Sections \ref{sec:maxrank} and \ref{sec:maxsnr} which are able to detect the signal under unknown variance condition. In this experiment we choose $N=1000$ and $K=10$. The elements of matrix $\mH$ are generated independently from $\cN(0,1/\sqrt{K})$ such that the expected value of the squared second norm of its rows is one. Signal $\vx$ and noise $\vw$ are both drawn from zero-mean Gaussian distribution with variances $\sigma_x^2$ and $\sigma_0^2$, respectively. We define the SNR as $10 \log_{10} \frac{\sigma_x^2}{\sigma_0^2}$ (in dB). The number of measurement devices is $M \in \{2,4,8\}$ and we choose $M_1=M_2=M/2$. The results are obtained from 10000 random trials  to gather sufficient statistics for reliable evaluation of the performance.

Figure \ref{fig:pd1} illustrates the performance in terms of probability of detection as a function of $\frac{\sigma_x^2}{\sigma_0^2}$.  The curves have been depicted for two cases of $\Nb=5$ and $\Nb=50$. As it can be observed, while Maximally-Uncorrelated compressive detector provides better performance for smaller $\Nb$ at high SNRs, it becomes inferior to Fully-Correlated detector for higher $\Nb$'s and at lower SNRs. This is also clear from Figure \ref{fig:roc1} where the performance in terms of Receiver Operating Characteristics (ROC) has been depicted. As we can see, again Maximally-Uncorrelated compressive detector outperforms Fully-Correlated method for high SNRs and small $\Nb$'s, while Fully-Correlated method works better in small SNRs and for higher number of measurement devices.

\begin{figure}[t!]
\begin{center}
\includegraphics[width=0.7\textwidth,height=0.5\textwidth]{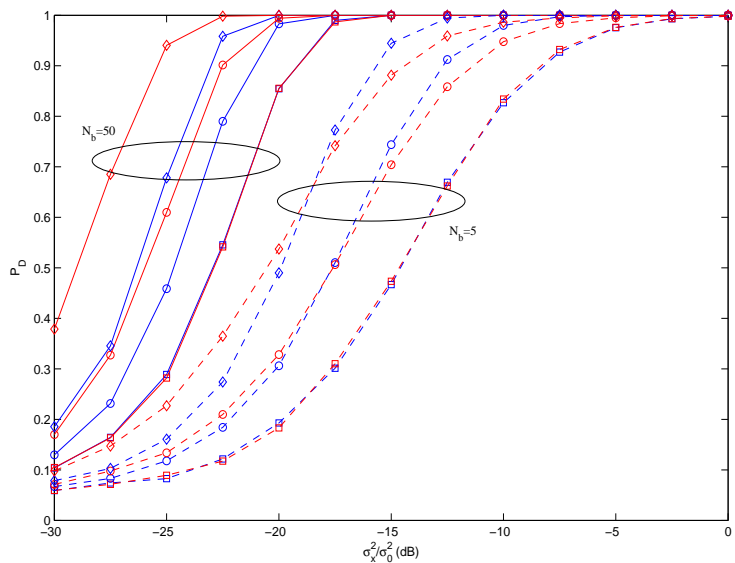}
\caption{Probability of Detection as a function of SNR for Maximally-Uncorrelated compressive detector (Blue) versus Fully-Correlated method (Red) when $\pfa=0.05$. Solid curves represent the case when $N_b=50$, and dashed curves represent the case when  $N_b=5$.  The number of measurement devices is $M=2$ (square), $M=4$ (circle), and $M=8$ (diamond). Furthermore we choose $M_1=M_2=M/2$.}
\label{fig:pd1}
\end{center}
\end{figure}

\begin{figure}[h!]
\begin{center}
\includegraphics[width=0.7\textwidth,height=0.5\textwidth]{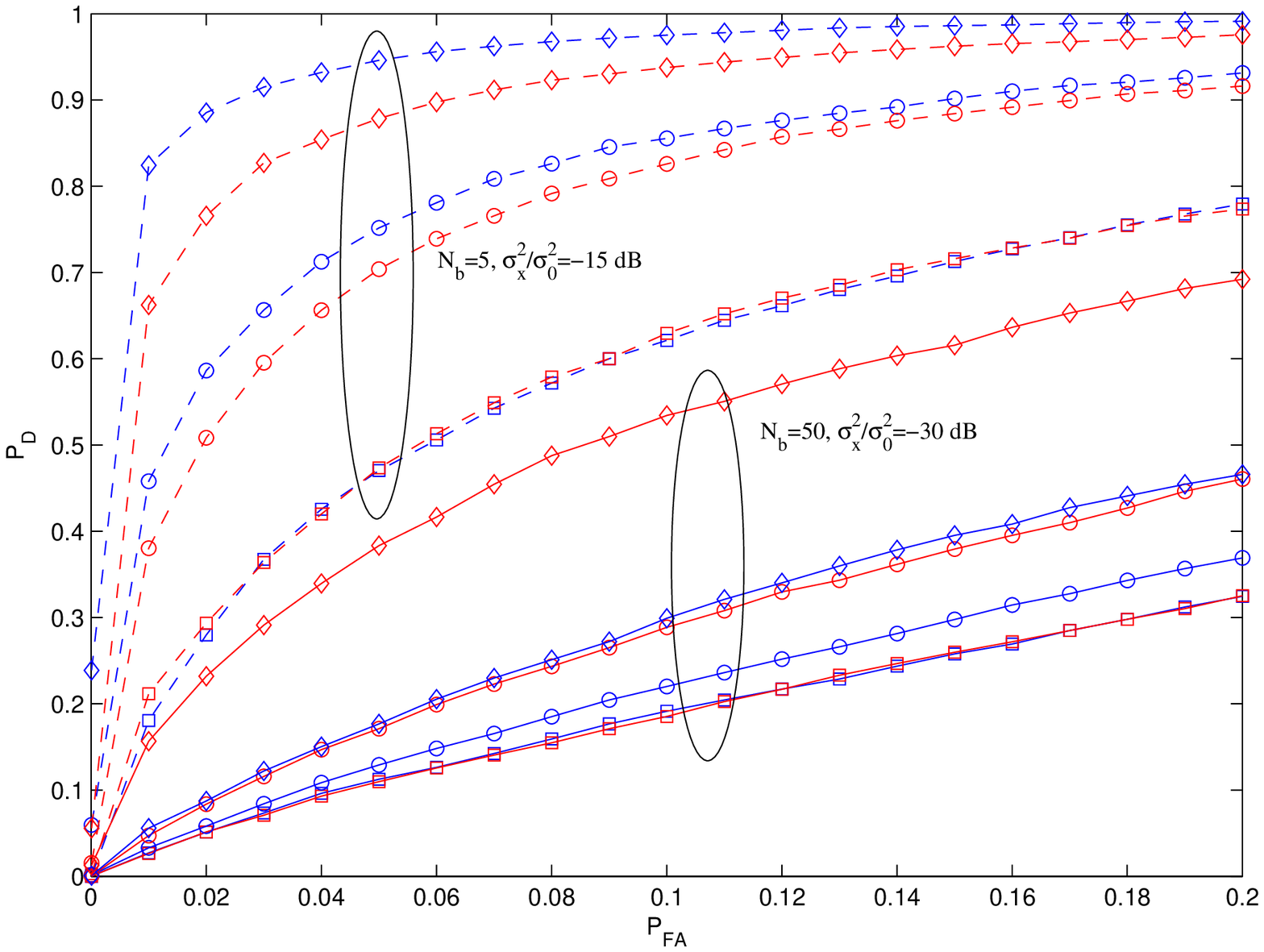}
\caption{ROC for Maximally-Uncorrelated compressive detector (Blue) versus Fully-Correlated method (Red). Solid curves represent the case when $N_b=50$ and ${\sigma_x^2}/{\sigma_0^2}=-30$ dB, and dashed curves represent the case when  $N_b=5$ and ${\sigma_x^2}/{\sigma_0^2}=-15$ dB.  The number of measurement devices is $M=2$ (square), $M=4$ (circle), and $M=8$ (diamond). Furthermore we choose $M_1=M_2=M/2$.}
\label{fig:roc1}
\end{center}
\end{figure}

The second experiment is to study the performance loss caused by the lack of knowledge about the noise variance. Figures \ref{fig:known1} and \ref{fig:known2} show the performance of the proposed Maximally-Uncorrelated detection method for two cases: when the noise variance is known and when it is unknown. As we expect, the lack of knowledge of noise variance results in considerable performance degradation.  However, reasonable performance is still obtained.

\begin{figure}[t!]
\begin{center}
\includegraphics[width=0.7\textwidth,height=0.5\textwidth]{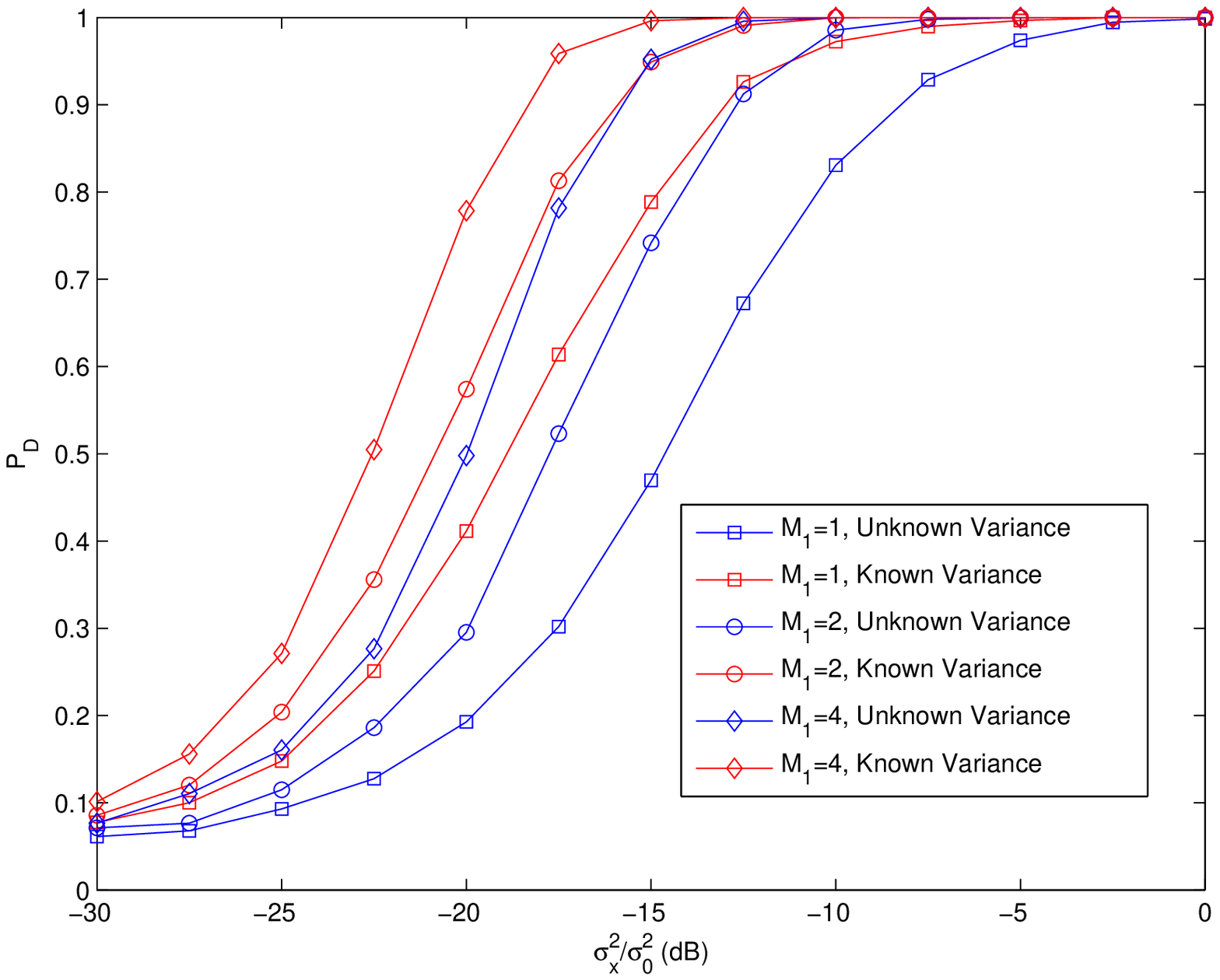}
\caption{Probability of Detection as a function of SNR for Maximally-Uncorrelated compressive detector under unknown variance (blue) versus known variance (red) when $\pfa=0.05$. The number of measurements over time is  $N_b=5$. The number of rows of $\mPhi_s$ is $M_1=1$ (square), $M_1=2$ (circle), and $M_1=4$ (diamond). When the variance is known we have $M=M_1$ and when it is unknown we choose $M_2=M_1$ and therefore we have $M=2 M_1$.}
\label{fig:known1}
\end{center}
\end{figure}

\begin{figure}[th!]
\begin{center}
\includegraphics[width=0.7\textwidth,height=0.5\textwidth]{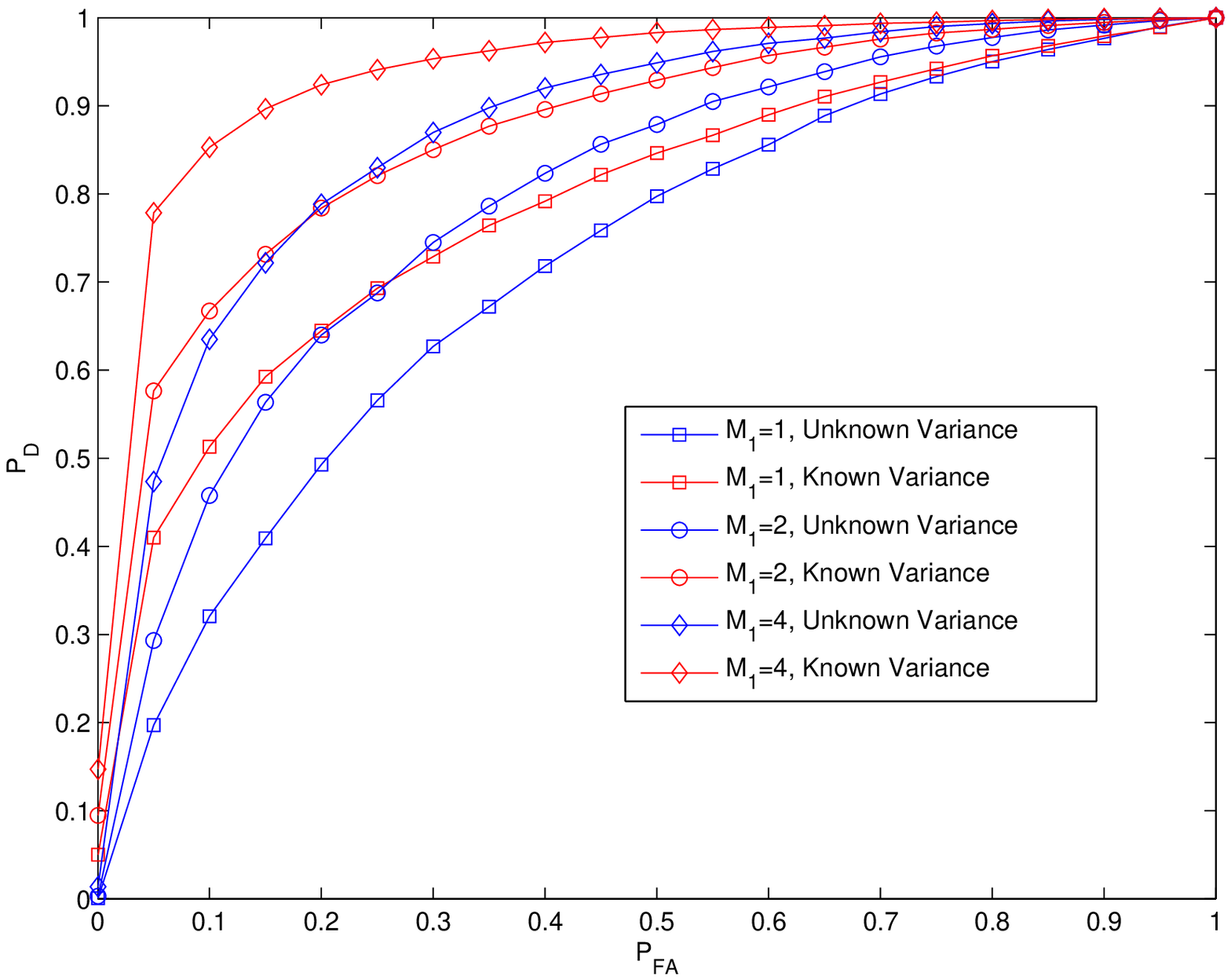}
\caption{ROC for Maximally-Uncorrelated detector for unknown variance (blue) versus known variance (red) when $N_b=5$ and ${\sigma_x^2}/{\sigma_0^2}=-20$ dB. The number of rows of $\mPhi_s$ is $M_1=1$ (square), $M_1=2$ (circle), and $M_1=4$ (diamond). When the variance is known we have $M=M_1$ and when it is unknown we choose $M_2=M_1$ and therefore we have $M=2 M_1$.}
\label{fig:known2}
\end{center}
\end{figure}

In the third experiment we study the effect of imprecise measurements and compensating it by deploying multiple measurement devices. The imprecision is modeled as described in Section \ref{subsec:imprecise}. Figure \ref{fig:roc2} illustrates the ROC curves for precise (red curves) and imprecise (blue curves) measurements when ${\sigma_x^2}=1$, and ${\sigma_0^2}=100$. The precision of all imprecise devices (blue curves) is chosen as $\delta^{-1}=100$. As expected and is clear from the comparison of the red curve with blue curve with $\nk=1$, in equivalent number of measurement devices, the imprecision in measurements results in a considerable loss in performance. However, this can be compensated by employing higher number of measurement devices. For instance, when $\nk=3$, the performance becomes equivalent to the precise case. This can be also verified by (\ref{eq:L}) in which the lower bound for performance improvement is $\nk=3$. When we go beyond this bound, e.g. $\nk=5$, the performance improves compared to the precise case.

Figure \ref{fig:pd2} shows how the probability of detection changes with the noise power for various precisions and hardware budgets of $\nk \in \{1,3\}$. The black curve depicts the performance with a single precise device. Here the performance enhancement with deploying $\nk=3$ times devices has been depicted. As it can be seen for a certain $\delta$, the performance of the case of $\nk=3$ is better than a single precise device as long as $\sigma_0^2<\frac{M}{\nk-1}\delta^{-1}$.  On the other hand,  when the variance of noise is higher than this limit, even employing $\nk$ times devices is not enough for compensating the imprecision.

\begin{figure}[t!]
\begin{center}
\includegraphics[width=0.7\textwidth,height=0.5\textwidth]{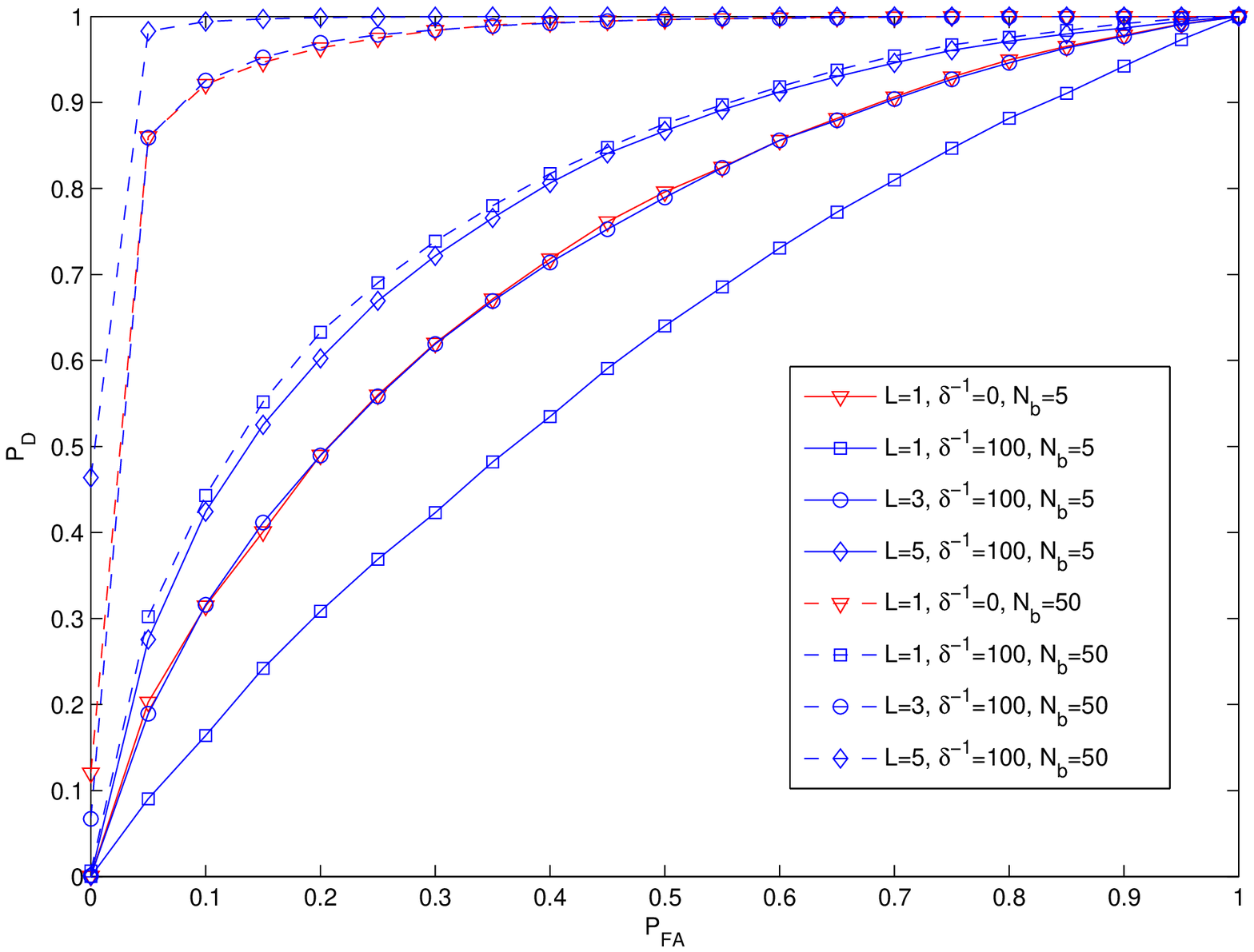}
\caption{ROC for imprecise measurements (blue curves) versus precise measurements (red curves). Solid curves represent the case when $N_b=5$ and dashed curves represent the case when  $N_b=50$. We have chosen ${\sigma_x^2}=1$, and ${\sigma_0^2}=100$.  The number of measurement devices in precise case is $M=2$ with $M_1=M_2=1$. For imprecise measurements, the device multiplicity has been chosen as $\nk \in \{1,3,5\}$.}
\label{fig:roc2}
\end{center}
\end{figure}

\begin{figure}[t!]
\begin{center}
\includegraphics[width=0.7\textwidth,height=0.5\textwidth]{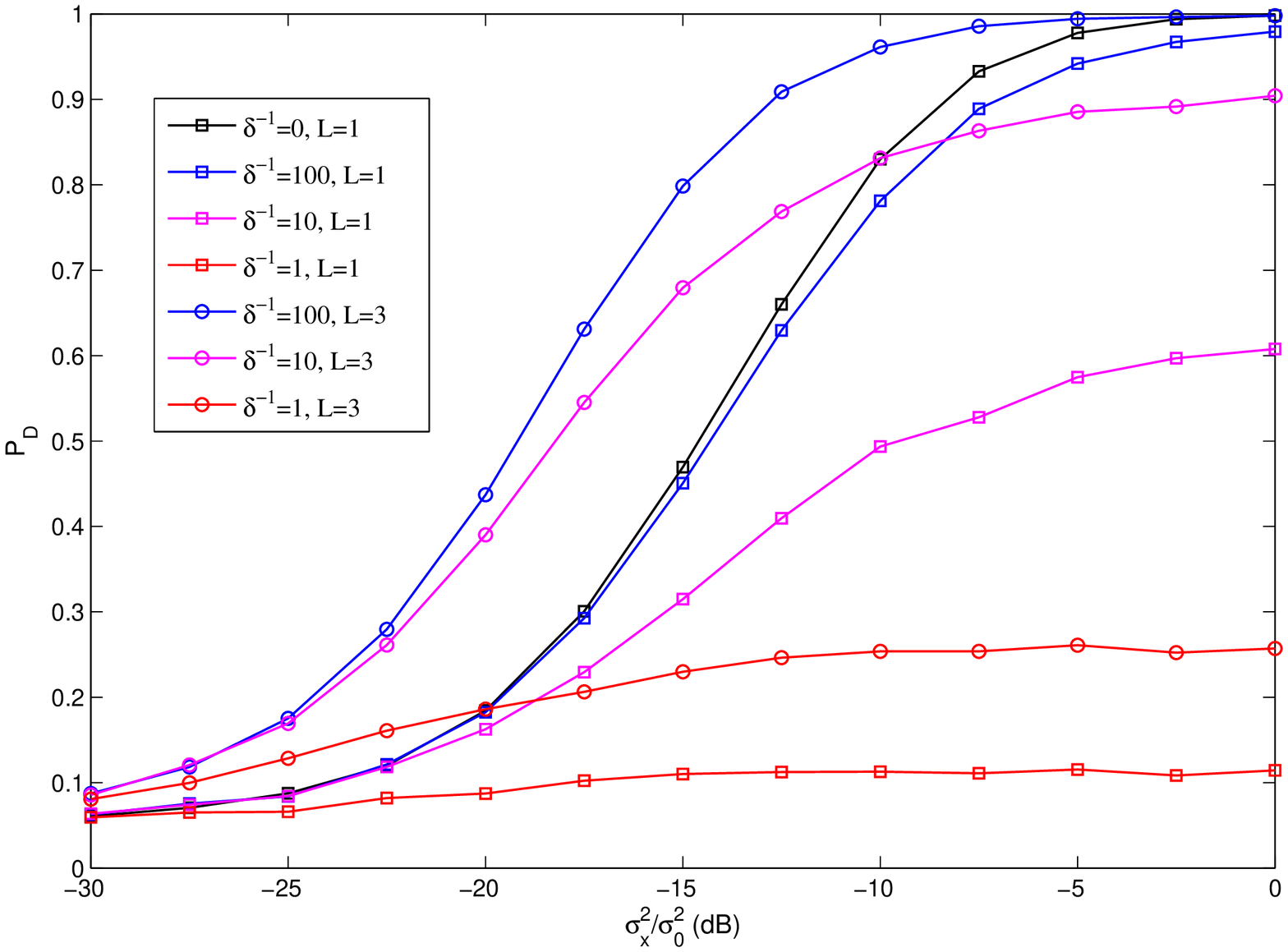}
\caption{Probability of detection for different precisions and device multiplicity values. The number of measurements per device is $N_b=5$. We have chosen ${\sigma_x^2}=1$.}
\label{fig:pd2}
\end{center}
\end{figure}

In the fourth experiment we study the performance of the proposed compressive detection method when the signal comes from a union of some subspaces with certain probabilities. For this experiment, the ambient dimension is $N=20$ and the signal dimension is $K=3$. In each random trial we first generate an $N \times N$ matrix with entries from a normal distribution with mean zero and variance $1/K$. The elements of union set $\sS$ in (\ref{eq:subspacesSet}) are then $Q$ different $K$-size combinations of columns of the generated matrix while the signal is generated from one of them which is chosen uniformly randomly from all $Q$ possible combinations. This implies that $\pi_1=\ldots=\pi_Q=1/Q$. The experiment is done for $Q\in\{1,10,50\}$. We remark that the case of $Q=1$ coincides with the basic case described in Section \ref{sec:2} as in this case there is no uncertainty over the subspace on which the signal lies. 

Figure \ref{fig:union1} illustrates the performance in terms of probability of detection for various values of $\sigma_x^2/\sigma_0^2$. The probability of false alarm here is set to $\pfa=0.05$ and the number of measurements is $\Nb=50$. The results are shown for $M \in \{2,4\}$ with $M_1=M_2=M/2$. Figure \ref{fig:union2} also shows the ROC curves when $\sigma_x^2/\sigma_0^2=-10 dB$. As it can be seen, as the cardinality of the subspace union set $\sS$ increases, the performance of the algorithm degrades. This is in turn because as the number of possible subspaces increases, the rows of measurement matrix $\mPhi$ are less matched to the singular vectors of the true subspace matrix.

\begin{figure}[t!]
\begin{center}
\includegraphics[width=0.7\textwidth,height=0.5\textwidth]{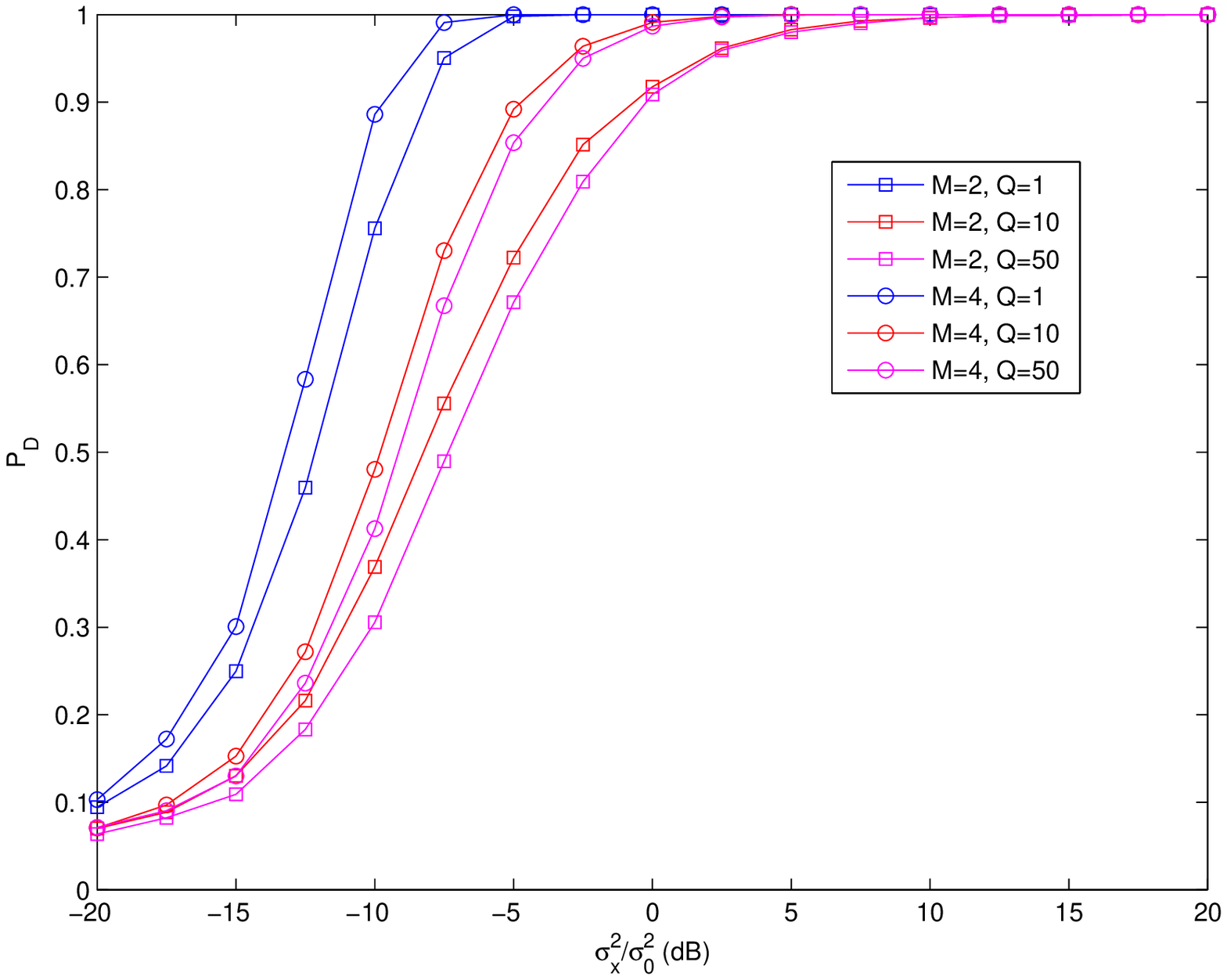}
\caption{Probability of Detection versus $\sigma_x^2/\sigma_0^2$ (in dB) for various numbers of measurement devices and cardinalities of the union sets. The number of measurements in time is $\Nb=50$.}
\label{fig:union1}
\end{center}
\end{figure}

\begin{figure}[t!]
\begin{center}
\includegraphics[width=0.7\textwidth,height=0.5\textwidth]{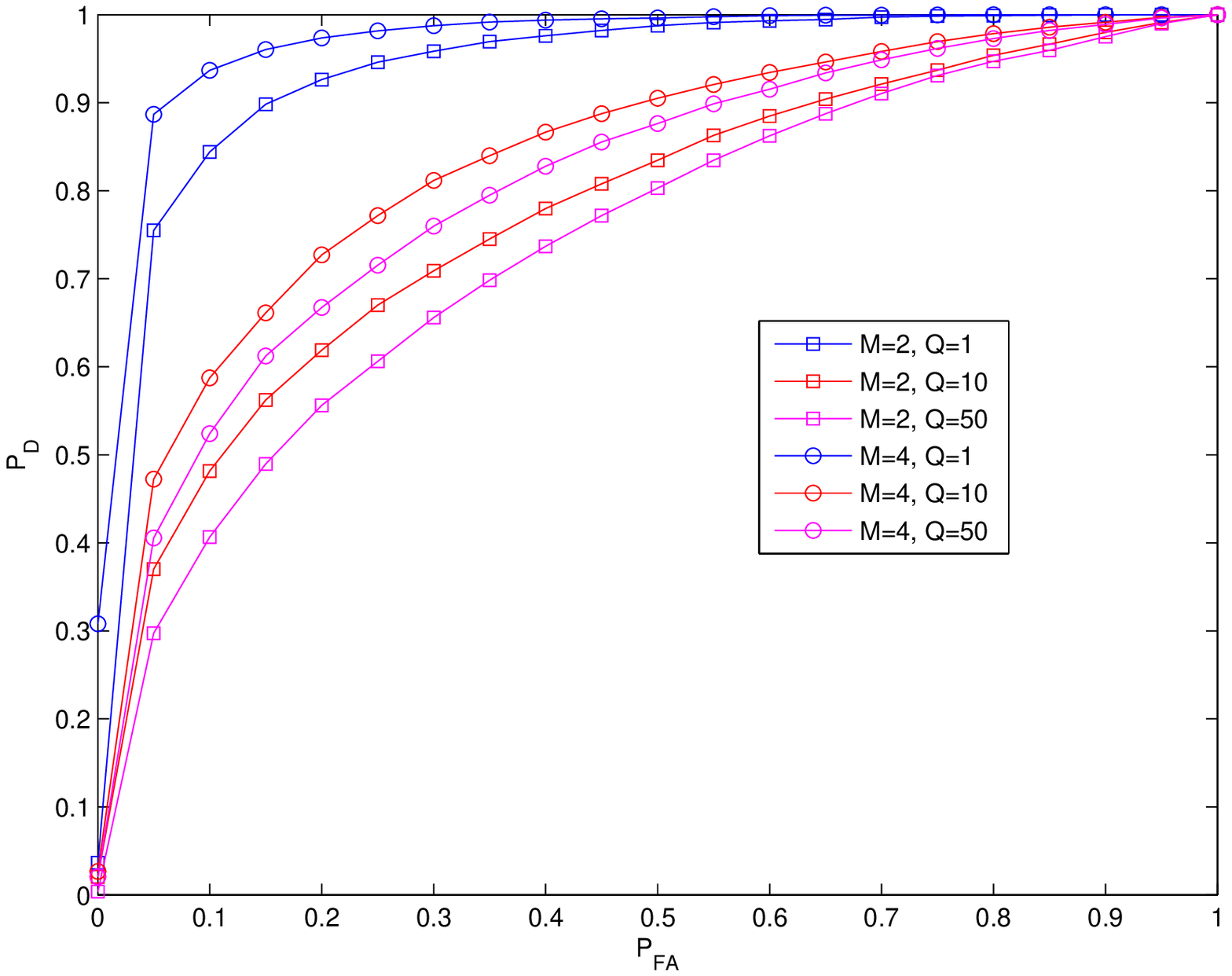}
\caption{ROC curves for various numbers of measurement devices and cardinalities of the union sets. The number of measurements is $\Nb=50$ and the SNR is $\sigma_x^2/\sigma_0^2=-10$ dB.}
\label{fig:union2}
\end{center}
\end{figure}

\section{Conclusion}
\label{sec:conclusion}
The problem of random signal detection from compressive measurements when the signal is lying (or leans toward) a low dimensional subspace was studied. Having the knowledge of the subspace structure, we proposed two measurement designs and formulated the hypothesis test and its performance metrics for each of them in the case that the noise variance is unknown. We also showed how the design can be simplified in the case that the noise variance is known. We analyzed the effects of imprecise measurements and showed how it can be compensated by deploying more identical measurement devices. The problem was also generalized for the case when signal belongs to a union of finite number of subspaces with known probabilities.

\section*{Acknowledgement}
This work was supported by the Finnish Funding Agency for Technology and Innovation (Tekes) under the project "Enabling Methods for Dynamic Spectrum Access and Cognitive Radio (ENCOR)", and the Academy of Finland under the project $\#251138$ "Digitally-Enhanced RF for Cognitive Radio Devices (DECORA)".



\end{document}